\documentclass[12pt]{article} 
\usepackage[sectionbib]{natbib}
\usepackage{array,epsfig,fancyheadings,rotating}
\usepackage[]{hyperref}  
\usepackage{sectsty, secdot}
\sectionfont{\fontsize{12}{14pt plus.8pt minus .6pt}\selectfont}
\renewcommand{\theequation}{\thesection\arabic{equation}}
\subsectionfont{\fontsize{12}{14pt plus.8pt minus .6pt}\selectfont}

\usepackage{amsmath}
\usepackage{amssymb}
\usepackage{amsfonts}
\usepackage{multirow}
\usepackage{amsthm}

\theoremstyle{definition}

\usepackage{amssymb, amsmath, amsthm, xcolor,float,graphicx,enumerate}
\usepackage{authblk}
\usepackage{setspace}
\usepackage{algorithm,epsf,subfig,algpseudocode}
\usepackage[cal=cm]{mathalfa}
\usepackage{natbib}
\usepackage{verbatim}
\newtheorem{lem}{Lemma}
\newtheorem{thm}{Theorem}

\newtheorem{cor}{Corollary}
 \DeclareMathOperator{\var}{var}
\DeclareMathOperator{\tr}{tr} 
\DeclareMathOperator{\IMSE}{IMSE}
\newcommand{\E}{\mathbb{E}}

\DeclareMathOperator{\se}{EA}

\DeclareMathOperator*{\argmin}{argmin}
\DeclareMathOperator*{\argmax}{argmax}

\DeclareMathOperator{\eff}{eff}
\DeclareMathOperator{\Eff}{Eff}
\DeclareMathOperator{\lse}{LEA}
\DeclareMathOperator{\mr}{Mm}
\DeclareMathOperator{\opt}{opt}
\DeclareMathOperator{\effcom}{eff-com}
\DeclareMathOperator{\phicom}{\Phi_p-com}

\DeclareMathOperator{\median}{median}
\newcommand{\reals}{\mathbb{R}}

\newcommand{\vx}{\boldsymbol{x}}

\newcommand{\vz}{\boldsymbol{z}}
\newcommand{\vf}{\boldsymbol{f}}
\newcommand{\vg}{\boldsymbol{g}}

\newcommand{\tphi}{\tilde{\Phi}}

\newcommand{\vbeta}{\boldsymbol{\beta}}

\newcommand{\mA}{{\mathsf A}}
\newcommand{\mB}{{\mathsf B}}

\newcommand{\mI}{{\mathsf I}}
\newcommand{\mM}{{\mathsf M}}

\newcommand{\mF}{{\mathsf F}}
\newcommand{\mG}{{\mathsf G}}
\newcommand{\mS}{{\mathsf S}}
\newcommand{\robdes}{{\xi^{\mr}_{\modelspace}}}
\newcommand{\optdesM}{{\xi^{\opt}_M}}
\newcommand{\optdesMj}{{\xi^{\opt}_{M_j}}}
\newcommand{\phicomdes}{{\xi^{\phicom}_{\modelspace}}}
\newcommand{\effcomdes}{{\xi^{\effcom}_{\modelspace}}}

\newcommand{\dif}{{\rm d}}
\newcommand{\vlambda}{\boldsymbol{\lambda}}

\newcommand{\modelspace}{\mathcal{M}'}

\begin{document}


\renewcommand{\baselinestretch}{2}

\markright{ \hbox{\footnotesize\rm Statistica Sinica
}\hfill\\[-13pt]
\hbox{\footnotesize\rm
}\hfill }

\markboth{\hfill{\footnotesize\rm YIOU LI AND LULU KANG AND XINWEI DENG} \hfill}
{\hfill {\footnotesize\rm A Maximin $\Phi_{p}$-Efficient Design for Multivariate GLM} \hfill}

\renewcommand{\thefootnote}{}
$\ $\par


\fontsize{12}{14pt plus.8pt minus .6pt}\selectfont \vspace{0.8pc}
\begin{center}
{\large \bf A Maximin $\Phi_{p}$-Efficient Design for Multivariate Generalized Linear Models}
\end{center}
\centerline{Yiou Li, Lulu Kang and Xinwei Deng} \vspace{.4cm} \centerline{\it
DePaul University, Illinois Institute of Technology, and Virginia Tech}
\vspace{.55cm} \fontsize{9}{11.5pt plus.8pt minus .6pt}\selectfont
\begin{quotation}
\noindent {\it Abstract:}
Experimental designs for a generalized linear model (GLM) often depend on the specification of the model, including the link function, the predictors, and unknown parameters, such as the regression coefficients.
To deal with uncertainties of these model specifications, it is important to construct optimal designs with high efficiency under such uncertainties.
Existing methods such as Bayesian experimental designs often use prior distributions of model specifications to incorporate model uncertainties into the design criterion.
Alternatively, one can obtain the design by optimizing the worst-case design efficiency with respect to uncertainties of model specifications.
In this work, we propose a new Maximin $\Phi_p$-Efficient (or Mm-$\Phi_p$ for short) design which aims at maximizing the minimum $\Phi_p$-efficiency under model uncertainties.
Based on the theoretical properties of the proposed criterion, we develop an efficient algorithm with sound convergence properties to construct the Mm-$\Phi_p$ design.
The performance of the proposed Mm-$\Phi_p$ design is assessed through several numerical examples.
\vspace{5pt}

\par
\end{quotation}\par

\def\thefigure{\arabic{figure}}
\def\thetable{\arabic{table}}

\renewcommand{\theequation}{\thesection.\arabic{equation}}

\fontsize{12}{14pt plus.8pt minus .6pt}\selectfont

\setcounter{section}{0} 
\setcounter{equation}{0} 

\lhead[\footnotesize\thepage\fancyplain{}\leftmark]{}\rhead[]{\fancyplain{}\rightmark\footnotesize\thepage}

\section{Introduction}\label{sec:intro}

Optimal design for generalized linear models (GLMs) \citep{sitter1995d-optimal,khuri2006design,silvey2013optimal, fedorov2013optimal} is an important topic in the design of experiments area.
In recent years, there have been new developments on both theoretical and algorithmic fronts, such as \cite{woods2011continuous,yang2011optimal, burghaus2014optimal, 14wu, waite2015designs, wong2019optimal} among many others.
A key challenge of optimal design for GLMs is that the design criterion often depends on the regression model assumption, including the specification of the link function, the linear predictor and the values of the unknown regression coefficients.
Many existing works focus on local optimal designs given a certain model specification, such as in \cite{09yang, li2009some, 12yang, 14wu,li2018efficient}.
Contrary to the local optimal design, one type of global optimal design takes the parameter uncertainty into consideration under two directions.
One direction is to consider a prior distribution of the unknown parameters, and construct the so-called Bayesian optimal design \citep{khuri2006design, amzal2006bayesian, woods2017bayesian}.
The design criterion is typically the integral of the local design criterion or efficiency with respect to the prior of the parameters.
When such integration is not analytically available, a standard solution is to sample from the prior distribution and use the weighted average of local design criteria or efficiencies as the objective function \citep{atkinson2015designs}.
Another direction is to use the minimax/maximin approach to minimize the design criterion or maximize the efficiency under the ``worst-case'' scenario.
\cite{sitter1992robust} introduced a minimax procedure for obtaining a design to deal with parameter uncertainty.
\cite{king2000minimax} proposed an efficient algorithm to construct a maximin design for the logistic regression model under D-optimality.
\cite{imhof2000graphical} developed an algorithm to maximize the minimum efficiency under two competing optimality criteria with a graphical method.
Note that existing literature on the maximin/minimax designs often focus on D-optimality and uncertainty of the unknown parameters.
The biggest challenge in maximin/minimax designs is that the design construction can be quite difficult \citep{atkinson2015designs}.

Besides the unknown parameters, there could be other uncertainties involved in a GLM, such as the specification of the link function and the linear predictor.
The literature on the designs for GLMs to deal with such kind of model uncertainty is relatively scarce.
\cite{woods2006designs} proposed a compromise design that minimizes the weighted average of the criteria, and each criterion is based on a potential model.
Later, \cite{dror2006robust} proposed using clustered local optimal designs, and showed the resulting design had a comparable performance with the compromise design through numerical examples.

In this work, we propose a new maximin $\Phi_p$-efficient design (denoted as Mm-$\Phi_p$) criterion for GLMs using the $\Phi_p$-efficiency \citep{75kiefer} and develop an efficient algorithm for design construction.
The proposed design, namely Mm-$\Phi_{p}$ design, can accommodate several types of uncertainties, including (i) uncertainty of the unknown parameter values; (ii) uncertainty of the linear predictor; and (ii) uncertainty of the link function.
Here, we focus on  \emph{approximate design} \citep{75kiefer,atkinson2014optimal}, which describes the design as a probability measure on a group of support points.
It provides the framework for us to investigate theoretical properties of the proposed design criterion, and pave a theoretical foundation to construct the efficient algorithm with desirable convergence properties.

The key idea of this work is to adopt a continuous and convex relaxation (i.e., the ``log-sum-exp'' approximation) as a tight approximation of the worst-case $\Phi_p$-efficiency with respect to uncertainty of model specifications.
With this relaxation, we arrive a tractable design criterion, which facilitates the theoretical investigation for developing an efficient algorithm to construct the corresponding design.
The merits of this idea is not restricted to the $\Phi_p$ criterion, even though $\Phi_p$ is already a quite general criterion including A-, D-, E-, and I-optimality criterion as special cases.
Through the demonstration of the proposed approach based on $\Phi_p$-criterion, it is apparent that this convex and smooth relaxation idea can be applied to other maximin design as long as the criterion is convex in the design.
The framework we have developed, including the general equivalence theorem and the design construction algorithm as well as its convergence, can be extended to other maximin design as well.

Other main contributions of this work are summarized as follows.
First, the proposed Mm-$\Phi_p$ design criterion is very general, covering various design criteria, such as D-, A-, E-optimality for estimation accuracy and I-, EI-optimality for prediction accuracy \citep{li2018efficient}.
Second, different from the Bayesian optimal design, the proposed Mm-$\Phi_p$ design is a maximin design, which avoids the choice of prior distributions on the model specifications.
Third, the proposed Mm-$\Phi_p$ design can flexibly accommodate the aforementioned three types of model uncertainties in GLM.
Finally, the proposed algorithm has impressive computational efficiency with sound theoretical properties, and can be easily modified to construct compromise designs and Bayesian optimal designs.

The rest of the article is organized as follows.
Section \ref{sec:crit+GET} describes the Mm-$\Phi_p$ design criterion and  investigates the theoretical properties.
In Section \ref{sec:algorithm}, an efficient algorithm is developed. 
Numerical examples are conducted in Section \ref{sec:examples} to examine the the performance of the proposed method.
We summarize the work with some discussions in Section \ref{sec:discussion}. All the technical proofs are detailed in the Appendix.

\section{The Mm-$\Phi_p$ Design Criterion and Its Properties}\label{sec:crit+GET}

Consider an experiment with $d$ design variables, $\vx = [x_1,...,x_d]$, and $x_j\in \Omega_j$, where $\Omega_j$ is a measurable domain of all possible values for $x_j$.
The experimental region, $\Omega$, is a certain measurable subset of $\Omega_1\times\cdots\times\Omega_d$.
For a GLM, the response $Y(\vx)$ is assumed to follow a distribution in the exponential family.
The link function, $h: \reals\rightarrow\reals$, provides the relationship between the linear predictor, $\eta=\vbeta^\top\vg(\vx)$, and $\mu(\vx)$, the mean of the response $Y(\vx)$ as
$\mu(\vx) = \E[Y(\vx)] = h^{-1}\left(\vbeta^\top\vg(\vx)\right)$,
where $\vg = [g_1,...,g_l]^\top$ are the known basis functions of the design variables, $\vbeta=[\beta_1,\beta_2,...,\beta_{l}]^\top$ are the corresponding regression coefficients parameters, and $h^{-1}$ is the inverse function of $h$.
The approximate design $\xi$ is defined as $\xi = \left\{\begin{array}{ccc}
\vx_1,&...,&\vx_n\\
\lambda_1,&...,&\lambda_n
\end{array}\right\}$,
where $\vx_1, \ldots \vx_{n}$ are the support points, and $0<\lambda_i<1$ represents the probability mass allocated to the corresponding support point $\vx_i$.
We use $M = (h,\vg,\vbeta)$ to denote the model specification of a GLM whose link function is $h$, basis functions are $\vg$, and the vector of the regression coefficients is $\vbeta$.
The Fisher information matrix of the GLM $M$ is
\begin{equation}\label{eqn:fisher}
\mI(\xi,M) = \sum\limits_{i=1}^n\lambda_i\vg(\vx_i)w(\vx_i,M)\vg^\top(\vx_i),
\end{equation}
where
$w(\vx_i,M) = \left[\var(Y(\vx_i))[h^{'}(\mu(\vx_i))]^2\right]^{-1}$.
Clearly, $\mI(\xi,M)$ depends all three components of $M=(h, \vg, \vbeta)$.
Various local optimal design criteria in the literature are all based on the Fisher information with a specified $M$.

\subsection{The Mm-$\Phi_p$ Design Criterion}\label{subsec:criterion}

To represent the uncertainties of a GLM, we denote the set of candidate link functions, the set of the candidate basis functions, and the domain of the regression coefficients as $\mathcal{H}$, $(\mathcal{G}|\mathcal{H})$, and ($\mathcal{B}|\mathcal{H},\mathcal{G})$, respectively.
The notation of conditioning presents the dependence of basis functions $\vg$ on the choice of link function $h$, and the dependence of regression coefficients $\vbeta$ on the choice of both $h$ and $\vg$.
The set $\mathcal{M} = \{M = (h,\vg,\vbeta): h\in\mathcal{H}, \vg\in(\mathcal{G}|\mathcal{H}), \vbeta\in (\mathcal{B}|\mathcal{H},\mathcal{G}) \}$ contains all model specifications of interest.

In the optimal design theory, \emph{efficiency} is a popular and scale-free performance measurement to compare the designs for a given criterion.
Specifically, for a generic design criterion $\Psi(\xi,\mathcal{M})$, which is to be minimized, the efficiency of a design $\xi$ relative to another design $\xi'$ is defined as \citep{06atk}
\begin{equation}\label{defn:minefficiency}
\eff_{\Psi}(\xi,\xi';\mathcal{M}) = \frac{\Psi(\xi',\mathcal{M})}{\Psi(\xi,\mathcal{M})}.
\end{equation}
Using such a definition of efficiency, the design $\xi$ is more efficient than design $\xi'$ as long as the efficiency in \eqref{defn:minefficiency}
is larger than 1.
When a single model specification is considered, i.e., $\mathcal{M} = \{M\}$, the criterion $\Psi$ becomes a local optimal design criterion.
When multiple specifications are considered, the criterion $\Psi$ corresponds to a global optimal design criterion, such as Bayesian optimality, compromise design optimality, minimax/maximin optimality, etc.

Throughout this work, for a specified model $M$, we use the generalized $\Phi_p$-optimality introduced in \cite{kiefer1974general}, which is
\begin{eqnarray}\label{eq:phi_p}
\Phi_p(\xi,M) = \left(q^{-1}\tr\left[\frac{\partial \vf(\vbeta)}{\partial \vbeta^{\top}}\mI(\xi,M)^{-1}\left(\frac{\partial \vf(\vbeta)}{\partial \vbeta^{\top}}\right)^{\top}\right]^p\right)^{1/p},\,\,\,0<p<\infty,
\end{eqnarray}
where $\vf(\vbeta) = [f_1(\vbeta),...,f_q(\vbeta)]^{\top}$ are some  functions of $\vbeta$.
Common examples are linear contrasts of the coefficients, such as $\beta_k$ and $\beta_j-\beta_{j'}$.
Note that the $\Phi_p$-optimality is essentially D-optimality as $p\rightarrow 0$ and E-optimality as $p\rightarrow\infty$.
Let us denote $\optdesM$ to be the local optimal design which minimizes the $\Phi_p$-criterion for model $M$.
According to \eqref{defn:minefficiency}, the $\Phi_p$-efficiency of any design $\xi$ relative to local optimal design $\optdesM$ given a specific $M = (h,\vg,\vbeta)$ is
\begin{eqnarray}\label{eqn:phieff}
\eff_{\Phi_p}(\xi,\optdesM;M) = \frac{\Phi_p(\optdesM,M)}{\Phi_p(\xi,M)}.
\end{eqnarray}
It is obvious that $0 \le \eff_{\Phi_p}(\xi,\optdesM;M) \le 1$ for any $\xi$, and the larger the $\Phi_p$-efficiency is, the more efficient the design $\xi$ is.
Under the idea of global maximin design, we consider the maximin $\Phi_p$-efficient design, which maximizes the smallest possible $\eff_{\Phi_p}(\xi,\xi^{\opt}_M;M)$ over all $M\in \mathcal{M}$.
That is, we consider a maximin design as
\begin{align}\label{eqn:orirobustdesign}
\xi^{*}&=\argmax\limits_{\xi} \inf_{M\in\mathcal{M}} \eff_{\Phi_p}(\xi,\optdesM;M).
\end{align}

In the optimization problem \eqref{eqn:orirobustdesign}, the infimum is used instead of minimum because it is not certain whether the minimum is attainable.
To simplify the problem, we take a closer look at the model set $\mathcal{M}$.
In practice, $\mathcal{H}$ usually contains a few potential link functions.
For example, the link function of Poisson regression for counting data is $h(\mu(\vx)) = \ln(\mu(\vx))$, and the link function of a GLM for binary data could be logistic function $h(\mu(\vx)) = \ln\left(\frac{\mu(\vx)}{1-\mu(\vx)}\right)$, or probit function $h(\mu(\vx)) = \Phi^{-1}(\mu(\vx))$, or a complementary log-log function $h(\mu(\vx)) = \ln(-\ln(1-\mu(\vx)))$.
The set of the candidate basis functions $(\mathcal{G}|\mathcal{H})$ is often finite too.
The typical basis functions used in GLMs are linear and/or higher-order polynomials of $\vx$.
Note that $(\mathcal{B}|\mathcal{H},\mathcal{G})$, the domain of $\vbeta$, often is uncountable since $\vbeta$ is considered to be continuous.
Consequently, the set $\mathcal{M}$ is an uncountable set, which may not ensure an attainable minimum.
A common remedy \citep{dror2006robust, woods2006designs,atkinson2015designs,woods2017bayesian} is to discretize $(\mathcal{B}|\mathcal{H, G})$ and create a finite and countable subset ($\mathcal{B}'|\mathcal{H},\mathcal{G})$.
The corresponding surrogate set $\mathcal{M}' = \{M = (h,\vg,\vbeta): h\in\mathcal{H}, \vg\in(\mathcal{G}|\mathcal{H}), \vbeta\in (\mathcal{B}'|\mathcal{H},\mathcal{G})\}$ is also a subset of the original $\mathcal{M}$.
Replacing $\mathcal{M}$ by $\mathcal{M'}$ in \eqref{eqn:orirobustdesign}, the solution of
\begin{eqnarray}\label{eqn:robustdesign}
\xi^{*} = \argmax \limits_{\xi} \min_{M\in\mathcal{M}'} \left[\eff_{\Phi_p}(\xi,\optdesM;M)\right]
\end{eqnarray}
is a sub-optimal solution of \eqref{eqn:orirobustdesign}.
When the discretization is adequate to form a close approximation of $\mathcal{M}$, the sub-optimal solution is expected to be close to the original optimal solution.

The design criterion in \eqref{eqn:robustdesign} is still a challenging optimization due to the non-smooth objective function $\min_{M\in\mathcal{M}'} \left[\eff_{\Phi_p}(\xi,\optdesM;M)\right]$
\citep{wong1992unified, wong1993heteroscedastic, schwabe1997maximin, king1998optimal, atkinson2015designs}.
We consider using ``Log-Sum-Exp" as a tight and smooth approximation to the minimum function, which is widely used in machine learning \citep{calafiore2014optimization}.
With the ``Log-Sum-Exp", one can have
\begin{align} \label{eq:lsebounds}
& \left[\ln\left(\sum\limits_{j=1}^m \exp\left(\frac{1}{\eff_{\Phi_p}(\xi,\optdesMj;M_j)}\right) \right)\right]^{-1} \le \min_{M\in\mathcal{M}'} \eff_{\Phi_p}(\xi,\optdesM;M)  \nonumber \\
& \leq \left[\ln\left(\sum\limits_{j=1}^m \exp\left(\frac{1}{\eff_{\Phi_p}(\xi,\optdesMj;M_j)}\right)\right)-\ln(m)\right]^{-1},
\end{align}
where $m$ is the cardinality of $\mathcal{M}'$, i.e., the number of potential model specifications in $\mathcal{M}'$.
The equality in the first inequality is obtained when $m=1$, and the equality in the second inequality holds when $\eff_{\Phi_p}(\xi,\optdesMj;M_j)$ remains the same for all $M_j\in\mathcal{M}'$.
Thus
maximizing $\left [ \ln\left(\sum\limits_{j=1}^m \exp\left(\frac{1}{\eff_{\Phi_p}(\xi,\optdesMj;M_j)}\right)\right)\right ]^{-1}$ leads to maximizing both the lower and upper bound of the worst (or the smallest) $\Phi_p-$efficiency.
Therefore, instead of solving \eqref{eqn:robustdesign}, which involves an inner minimization of $\Phi_p$-efficiency,
we propose to use the ``Log-Sum-Exp" approximation of the worst-case $\Phi_p$-efficiency as the design criterion, which is to minimize
\begin{equation}\label{eqn:lse}
\lse(\xi,\mathcal{M}') \triangleq   \ln\left(\sum\limits_{j=1}^m \exp\left(\frac{1}{\eff_{\Phi_p}(\xi,\optdesMj;M_j)}\right)\right).
\end{equation}
Minimizing $\lse(\xi,\mathcal{M}')$ is the same as maximizing $\left [ \ln\left(\sum\limits_{j=1}^m \exp\left(\frac{1}{\eff_{\Phi_p}(\xi,\optdesMj;M_j)}\right)\right)\right ]^{-1}$
since $\ln\left(\sum\limits_{j=1}^m \exp\left(\frac{1}{\eff_{\Phi_p}(\xi,\optdesMj;M_j)}\right)\right)>0$.
We call the $\lse(\xi,\modelspace)$, which aims at maximizing the minimal $\Phi_p$-efficiency, the Mm-$\Phi_p$ criterion.
The design that minimizes $\lse(\xi,\modelspace)$ is called the Mm-$\Phi_p$ design for the  surrogate model set $\modelspace$, denoted by $\robdes$.

It is obvious that minimizing $\lse(\xi,\mathcal{M}')$ is equivalent to minimizing
\begin{align}\label{eqn:criterion}
\se(\xi,\modelspace) = \triangleq \sum\limits_{j=1}^m \exp\left(\frac{1}{\eff_{\Phi_p}(\xi,\optdesMj;M_j)}\right)
                      = \sum\limits_{j=1}^m \exp\left(\frac{\Phi_p(\xi,M_{j})}{\Phi_p(\optdesMj,M_{j})}\right).
\end{align}
That is to say $\robdes = \argmin\limits_{\xi} \se(\xi,\modelspace)=\argmin\limits_{\xi} \lse(\xi,\modelspace)$.

In Section \ref{subsec:robvscom} and \ref{subsec:theory}, we first compare the proposed maximin $\Phi_p$-efficient design $\robdes$ with the well-known compromise design in \cite{woods2006designs}
and then show the convexity of $\se(\xi,\modelspace)$ with respect to $\xi$, as well as the necessary and sufficient conditions of the Mm-$\Phi_p$ design $\robdes$.

\subsection{Connection to Compromise Design}\label{subsec:robvscom}
\cite{woods2006designs} proposed a compromise design that optimizes the weighted average of certain  criteria, where each criterion is based on a potential model from some prior.
It means that the compromise design requires a prior distribution $p(M)$ for the model specifications $M\in\mathcal{M}'$.
The prior distribution can be as simple as a uniform distribution or other informative distributions.

There can be two different ways to define a compromise design.
The first way aims at maximizing a weighted average of the local $\Phi_p$-efficiencies.
That is
\[
\effcomdes = \argmax_{\xi} \sum_{j=1}^m p(M_j)\eff_{\Phi_p}(\xi,\optdesMj;M_j),
\]
and it is henceforth called the eff-compromise design.
Clearly, this averaged local efficiencies is not smaller than the reciprocal of $\lse(\xi,\modelspace)$ since
\[
[\lse(\xi,\modelspace)]^{-1}\leq \min_{M\in\mathcal{M}'} \eff_{\Phi_p}(\xi,\optdesM;M)\leq  \sum_{j=1}^m p(M_j)\eff_{\Phi_p}(\xi,\optdesMj;M_j).
\]
Thus the compromise design maximizes an upper bound of the worst $\Phi_p$-efficiency.
This is not as ideal as $\lse(\xi,\modelspace)$.
Minimizing $\lse(\xi,\modelspace)$ simultaneously maximizes a lower and an upper bound of the worst $\Phi_p$-efficiency (see \eqref{eq:lsebounds}),
even though the two upper bounds $\left[\lse(\xi,\mathcal{M}')-\ln(m)\right]^{-1}$ and $\sum_{j=1}^m p(M_j)\eff_{\Phi_p}(\xi,\optdesMj;M_j)$ can be both attainable, depending on the prior distributions.

Another type of compromise design is to minimize the weighted average of local $\Phi_p$-criterion. That is
\[
\phicomdes = \argmin_{\xi} \sum_{j=1}^m p(M_j)\Phi_p(\xi, M_j),
\]
which is henceforth called the $\Phi_p$-compromise design.
Such a design criterion is more consistent with the classic Bayesian optimal design.
According to \cite{woods2006designs} and \cite{atkinson2015designs}, the Bayesian optimal design can be considered as a special case of the compromise design, as the former only deals with the uncertainty of the unknown parameters of the GLMs, whereas the compromise design handles all three kinds of uncertainties that are listed previously, including uncertainty of the parameters.
We would like to point out that the $\Phi_p$-compromise design can be sensitive to the choice of the prior distribution, especially when the optimal criterion values of different model specifications are very different. On the contrary, $\lse(\xi,\modelspace)$ does not assume any prior distribution and is robust to all the choices of the prior distribution of model specifications.

\subsection{General Equivalence Theorem}\label{subsec:theory}

To develop an efficient algorithm to construct the Mm-$\Phi_p$ design, we study the convexity of the objective function $\se(\xi,\modelspace)$ with respect to $\xi$, and summarize the necessary and sufficient conditions of the Mm-$\Phi_p$ design $\robdes$ in a General Equivalence Theorem.
To make this part concise, we list the major results here and place the lemmas and the proofs in the Appendix.

For a model specification $M_j\in \mathcal{M}'$, we simplify the notation of the information matrix $\mI(\xi,M_j)$ to be $\mI_j(\xi)$, the weight function $w(\vx,M_j)$ in \eqref{eqn:fisher} to be $w_j(\vx)$, the $\Phi_p$-criterion value of a design  $\Phi_p(\xi, M_j)$ to be $\Phi_p^j(\xi)$, and the $\Phi_p$-criterion value $\Phi_p(\optdesMj, M_j)$ of the local optimal design to be $\Phi_p^{\opt_j}$.
Then, we can rewrite $\se(\xi,\modelspace)$ as $\se(\xi,\modelspace)=\sum\limits_{j=1}^m \exp\left(\frac{\Phi_p^j(\xi)}{\Phi_p^{\opt_j}}\right)$.
Lemma \ref{lem:SeConvex} in the Appendix proves the convexity of $\se(\cdot,\modelspace)$ with respect to $\xi$.
Given two designs $\xi$ and $\xi'$, the directional derivative of $\se(\xi,\modelspace)$ in the direction of $\xi'$ is defined as follows.
\begin{equation}\label{defn:dirder}
\nabla_{\xi'}\se(\xi,\modelspace):=\phi(\xi',\xi) = \lim\limits_{\alpha\rightarrow 0^+}\frac{\se((1-\alpha)\xi+\alpha \xi',\modelspace)-\se(\xi,\modelspace)}{\alpha}, \quad \alpha\in [0,1].
\end{equation}
Lemma \ref{lem:dirder} in the Appendix derives the concrete formula of $\phi(\xi',\xi)$.
If $\xi'$ only contains a single support point $\vx$ with corresponding weight $\lambda=1$, the directional derivative of $\se(\xi,\modelspace)$ in the direction of $\xi'$ is a special case of Lemma \ref{lem:dirder}.
We denote this directional derivative as $\phi(\vx,\xi)$, and give its formula in Lemma \ref{lem:dirder2} in the Appendix.
Following Lemma \ref{lem:dirder2}, we also provide the specific formulas of $\phi(\vx,\xi)$ for the D-, A- and EI-optimality.
With these results, we can obtain the General Equivalence Theorem \ref{thm:equi_thm} for the Mm-$\Phi_p$ design that maximizes $\lse(\xi,\modelspace)$, or equivalently minimizes $\se(\xi,\modelspace)$.

\begin{thm}[General Equivalence Theorem]\label{thm:equi_thm}
The following three conditions of a design $\robdes$ are equivalent:
\begin{enumerate}
\item
The design $\robdes$ minimizes $\lse(\xi,\modelspace)$ and $\se(\xi,\modelspace)$.
\item
$\phi(\vx,\robdes)\geq 0$ holds for any $\vx\in\Omega$, and the inequality becomes equality if $\vx$ is a support point of the design $\robdes$.
\end{enumerate}
\end{thm}


The General Equivalence Theorem \ref{thm:equi_thm} for the $\lse$ criterion in \eqref{eqn:lse} provides important guidelines on how the support points of the Mm-$\Phi_p$ design should be added in a sequential manner.
The proposed algorithm for the Mm-$\Phi_p$ design (detailed in Section \ref{sec:algorithm}) iterates between adding the support point and updating the weights $\lambda_i$'s, i.e., which can be considerd as a Fedorov-Wynn type of algorithm \citep{dean2015handbook}.
In each step of an iteration, we add one design point $\vx^*$ into the current design as a support point, if $\vx^*$ meets the following two conditions.
The first condition is that its directional derivative is negative, $\phi(\vx^*,\xi)<0$.
Otherwise, if there does not exist an $\vx\in \Omega$ such that $\phi(\vx,\xi)<0$, then $\xi$ already reaches the optimal.
The second condition is that the directional derivative $\vx^*$ reaches the minimum, or the \emph{size} of the directional derivative is maximal compared to other possible points whose directional derivative values are also negative.
This condition leads to the maximum reduction of $\se(\xi,\modelspace)$ if $\vx^*$ is added to $\xi$.

After the design point $\vx^*$ is added, the weights of all design points in the current design need to be updated.
Thus, it is important to investigate the property of the optimal weights when the design points are given.
Given design points $\vx_1,\vx_2,...,\vx_n$, the weight vector $\vlambda = [\lambda_1,\lambda_2,\ldots,\lambda_n]^\top$ is the only variable for the design.
We emphasize this by adding a superscript $\vlambda$ in the notation of the design and denote it as
$\xi^{\vlambda} = \Big\{\begin{array}{ccc}
\vx_1,&...,&\vx_n\\
\lambda_1,&...,&\lambda_n
\end{array}\Big\}$.
Consider $\se(\xi^{\vlambda},\modelspace)$ as a function of $\vlambda$, i.e.,
\begin{eqnarray}\label{eqn:weightlb}
\se(\cdot,\modelspace): \{\vlambda = (\lambda_1,\cdots,\lambda_n): \lambda_i > 0,\sum\lambda_i=1\}\mapsto \sum\limits_{j=1}^m \exp\left(\frac{\Phi_p(\xi^{\vlambda},M_j)}{\Phi_p^{\opt_j}}\right).
\end{eqnarray}
The optimal weight vector $\vlambda^*$ should be the one that minimizes $\se(\xi^{\vlambda},\modelspace)$ with the given support points $\vx_1,...,\vx_n$.
Lemma \ref{lem:ConvWeights} in the Appendix proves the convexity of $\se(\xi^{\vlambda},\modelspace)$ with respect to $\vlambda$.
Corollary \ref{thm:equi_weight} provides a sufficient and necessary condition on the optimal weights for a design whose support points are fixed.
It is a special case of Theorem \ref{thm:equi_thm} when the experimental region is restricted to the set $\Omega = \{\vx_1,...,\vx_n\}$.
\begin{cor}[Conditions of Optimal Weights]\label{thm:equi_weight}
Given a set of design points $\vx_1,...,\vx_n$, the following three conditions on the weight vector $\vlambda^* = [\lambda^*_1,...,\lambda^*_n]^\top$ are equivalent:
\begin{enumerate}
\item
The weight vector $\vlambda^*$ minimizes $\lse(\xi^{\vlambda},\modelspace)$ and $\se(\xi^{\vlambda},\modelspace)$.
\item
 For all $\vx_i$, with $\lambda_i^*>0$, $\phi(\vx_i,\xi^{\vlambda^*}) = 0;$
for all $\vx_i$ with $\lambda_i^*=0$, $\phi(\vx_i,\xi^{\vlambda^*}) \geq 0.$
\end{enumerate}
\end{cor}

\section{Efficient Algorithm of Constructing Mm-$\Phi_{p}$ Design}\label{sec:algorithm}
This section details the proposed sequential algorithm, named as \textbf{Mm-$\Phi_{p}$ Algorithm}, to construct the Mm-$\Phi_p$ design $\robdes$.
The proposed algorithm has a sound theoretical rationale as well as impressive computational efficiency.
The key idea of the proposed algorithm is as follows.
In each sequential iteration, a new design point $\vx^*$ with the smallest negative value of directional derivative $\vx^* = \argmin\limits_{\vx} \phi(\vx,\xi)<0$ is added to the current design,
and then the Optimal-Weight Procedure (detailed in Section \ref{sec: weight updating}) is used to optimize the weights of the current design points.
The stopping rule of the proposed sequential algorithm can be determined based on the efficiency of the obtained design.

The proposed \textbf{Mm-$\Phi_{p}$ Algorithm}, following a similar spirit as the sequential Wynn-Fedorov type algorithm, is to add the new design point after the optimal weights of the existing design points are achieved.
In each iteration, the design point that minimizes the directional derivative $\phi(\vx,\xi)$ will be added into the design to gain a maximum reduction of $\se$ criterion value.
Then, the weights of all design points in the current design are optimized, which will be described in Section \ref{sec: weight updating}.
Theoretically, the algorithm should terminate until the directional derivatives of all candidate design points in the experimental region are nonnegative.
However, this stopping rule is impractical since it requires many iterations to make all the directional derivative values strictly positive (numerically it is unlikely to have exactly zero cases).
To address this issue, we consider terminating the algorithm when the design efficiency is large enough, say close to 1.
Such a stopping criterion is much better than terminating the algorithm when the directional derivative $\min\limits_{\vx\in\Omega}\phi(\vx,\xi)>\epsilon$ with a small negative $\epsilon$.
The drawback of the later rule is that the choice of $\epsilon$ does not directly reflect the quality of the achieved design, since $\phi(\vx,\xi)$ is the directional derivative.

Following the general definition of design efficiency in \eqref{defn:minefficiency}, we denote the efficiency of a design $\xi$ relative to the Mm-$\Phi_p$ design $\robdes$ that minimizes the Mm-$\Phi_p$ criterion $\lse$ as:
\begin{equation}\label{eqn:robusteff}
\Eff_{\lse}(\xi,\robdes;\modelspace) = \frac{\lse(\robdes,\modelspace)}{\lse(\xi,\modelspace)}.
\end{equation}
Since $\Eff_{\lse}(\xi,\robdes;\modelspace)$ involves $\robdes$, which is unknown, we derive a lower bound of it in Theorem \ref{thm:lowerboundeff}.
Instead of using $\Eff_{\lse}(\xi,\robdes;\modelspace)$ as the stopping rule, we can use the lower bound of $\Eff_{\lse}(\xi,\robdes;\modelspace)$ as the stopping rule.

\setcounter{lem}{4}
\begin{lem}\label{lem:ineqofdirder}
For any design $\xi$ and the Mm-$\Phi_p$ design $\robdes$ that minimizes $\se(\xi,\modelspace)$ or equivalently minimizes $\lse(\xi,\modelspace)$, the following inequality holds:
$$\min\limits_{\vx\in\Omega} \phi(\vx,\xi)\leq \phi(\robdes,\xi)\leq \se(\robdes,\modelspace)-\se(\xi,\modelspace)\leq 0,$$
where $\phi(\vx,\xi)$ and $\phi(\robdes,\xi)$  are the directional derivatives defined in \eqref{defn:dirder}.
\end{lem}

\begin{thm}[A Lower Bound of $\lse$-Efficiency]\label{thm:lowerboundeff}
Design $\robdes$ is the Mm-$\Phi_p$ design that minimizes $\lse$ criterion in \eqref{eqn:lse}.
The $\lse$-efficiency defined in \eqref{eqn:robusteff} of any design $\xi$ relative to $\robdes$ is bounded below by
$$\Eff_{\lse}(\xi,\robdes;\modelspace) \geq  1+2\frac{\min\limits_{\vx\in\Omega} \phi(\vx,\xi)}{\se(\xi,\modelspace)}.$$
\end{thm}

Using the lower bound of $\lse$-efficiency in Theorem \ref{thm:lowerboundeff} as the stopping criterion,
the proposed algorithm terminates when the lower bound  $1+2\frac{\min\limits_{\vx\in\Omega} \phi(\vx,\xi)}{\se(\xi,\modelspace)}$ exceeds a user-specified value, $Tol_{\text{eff}}$.
Here $Tol_{\text{eff}}$ should be set close to 1, say $Tol_{\text{eff}} = 0.99$, or equivalently $\frac{\min\limits_{\vx\in\Omega} \phi(\vx,\xi)}{\se(\xi,\modelspace)}\geq -0.005$.
With this stopping rule, the sequential algorithm to construct the Mm-$\Phi_p$ design is described in Algorithm \ref{alg:sequential}.
The $MaxIter_2$ is the maximum number of iterations allowed of adding design points, and we set it to be 200.

\begin{algorithm}
  \caption{ (\textbf{Mm-$\Phi_p$ Algorithm}) The Sequential Algorithm for Mm-$\Phi_p$ Design. \label{alg:sequential}}
  \begin{algorithmic}[1]
  \State For each model specification $M_j\in \modelspace$, construct the local optimal design and calculate the corresponding optimatlity criterion value $\Phi_p^{\opt_j}$.
  \State Generate an $N$ points candidate pool $\mathcal{C}$.
  \State Choose an initial design points set $\mathcal{X}^{(0)} = \left\{\vx_1,\cdots,\vx_{l+1}\right\}$ containing $l+1$ points.
  \State Obtain optimal weights $\vlambda^{(0)}$ of initial design points set $\mathcal{X}^{(0)}$ using Algorithm \ref{alg:weight} (\textbf{Optimal-Weight Procedure}) and form the initial design $\xi^{(0)} = \left\{\begin{array}{cc} \mathcal{X}^{(0)}\\ \vlambda^{(0)}
  \end{array}\right\}$.
  \State Calculate the lower bound of $\lse$-efficiency of $\xi^{(0)}$:
 \[\text{eff.low} = 1+2\frac{\min\limits_{\vx\in\mathcal{C}} \phi(\vx,\xi^{(0)})}{\se(\xi^{(0)},\modelspace)}.\]
  \State Set $r=1$.
  \While {$\text{eff.low}<Tol_{\text{eff}}$ and $r< MaxIter_1$}
  \State
 Add the point $\vx_r^* = \argmin \limits_{\vx \in \mathcal{C}} \phi(\vx,\xi^{(r-1)})$ to the current design points set, i.e., $\mathcal{X}^{(r)} = \mathcal{X}^{(r-1)}\cup \{\vx_r^*\}$,
  where $\phi(\vx,\xi^{(r)})$ is given in Lemma \ref{lem:dirder2}.
\State Obtain optimal weights $\vlambda^{(r)}$ of the current design points set $\mathcal{X}^{(r)}$ using Algorithm \ref{alg:weight} (\textbf{Optimal-Weight Procedure}) and form the current design $\xi^{(r)} = \left\{\begin{array}{cc} \mathcal{X}^{(r)}\\ \vlambda^{(r)}
  \end{array}\right\}$.
  \State Calculate the lower bound of $\lse$-efficiency of $\xi^{(r)}$,
\[\text{eff.low} = 1+2\frac{\min\limits_{\vx\in\mathcal{C}} \phi(\vx,\xi^{(r)})}{\se(\xi^{(r)},\modelspace)}.\]
\State $r=r+1$.
 \EndWhile
\end{algorithmic}
\end{algorithm}

In Section \ref{sec: convergence}, we provide some theoretical properties on the convergence of the Mm-$\Phi_p$ Algorithm.
Note that the \textbf{Mm-$\Phi_p$ Algorithm} requires optimizing the weights $\vlambda^{(r)}$ of the current design points in each sequential iteration.
Section \ref{sec: weight updating} describes the procedure on how to optimize the weight given the design points.

\subsection{Convergence of the Mm-$\Phi_p$ Algorithm}\label{sec: convergence}

The sequential nature of the proposed \textbf{Mm-$\Phi_p$ Algorithm} (i.e., Algorithm \ref{alg:sequential}) makes it efficient in computation as it adds one design point in each iteration.
Moreover, we can establish the theoretical convergence of Algorithm \ref{alg:sequential}, which is stated as follows.

\begin{thm}[Convergence of Algorithm \ref{alg:sequential}(Mm-$\Phi_p$ Algorithm)]\label{thm:cong-algo2}
Assume the candidate pool $\mathcal{C}$ contains all the support points of the Mm-$\Phi_p$ design $\robdes$.
The design constructed by Algorithm \ref{alg:sequential} converges to $\robdes$ that minimizes $\lse(\xi,\modelspace)$, i.e.,
\[\lim\limits_{r\rightarrow\infty} \lse(\xi^{(r)},\modelspace) = \lse(\robdes,\modelspace).\]
\end{thm}

Besides its theoretically guaranteed convergence property, Algorithm \ref{alg:sequential} also converges fast with no more than 50 iterations in all the numerical examples, although the maximal number of iteration is set to be 200. More details about the speed of convergence and computational time are reported in Section \ref{sec:examples}.

We would like to remark that, at the beginning of Algorithm \ref{alg:sequential},  the local optimal design and the corresponding optimality criterion value $\Phi_p^{\opt_j}$ need to be calculated for each model specification $M_j \in \modelspace$.
It is because they are involved in $\se(\xi, \modelspace)$ and all its derivatives.
However, we only need to compute them once. Using the algorithm proposed by \cite{li2018efficient}, we can construct local $\Phi_p$-optimal designs for GLMs efficiently with guaranteed convergence.

\subsection{An Optimal-Weight Procedure Given Design Points}\label{sec: weight updating}
Based on Corollary \ref{thm:equi_weight}, with a given set of design points $\vx_1,\cdots,\vx_n$, a sufficient condition that $\vlambda^*$ minimizes $\se(\xi^{\vlambda},\modelspace)$ is:
\[
\phi(\vx_i,\xi^{\vlambda^*}) = 0, \mbox{ for } i=1, \ldots, n,
\]
or equivalently (based on Lemma \ref{lem:dirder2}),
\begin{equation}\label{eqn:suffweight1}
\small
\left\{\begin{array}{rll}
q\sum\limits_{j=1}^m\tphi_0^j(\xi^{\vlambda^*}) &= \sum\limits_{j=1}^m \tphi_0^j(\xi^{\vlambda^*}) w_j(\vx_i)\vg_j^{\top}(\vx_i)\mM_j(\xi^{\vlambda^*})\vg_j(\vx_i), & p=0;\\
q^{1/p}\sum\limits_{j=1}^m\tphi_p^j(\xi^{\vlambda^*})\Phi_p^j(\xi^{\vlambda^*}) &= \sum\limits_{j=1}^m\tphi_p^j(\xi^{\vlambda^*})w_j(\vx_i)\left(\tr\left(\mF_j(\xi^{\vlambda^*})\right)^p\right)^{1/p-1}\vg_j^{\top}(\vx_i) \mM_j(\xi^{\vlambda^*})\vg_j(\vx_i),& p>0.
\end{array}\right.
\end{equation}
where
$\tphi^j_p(\xi) = \left[\Phi_p^{\opt_j}\right]^{-1}\exp\left(\frac{\Phi_p^j(\xi)}{\Phi_p^{\opt_j}}\right)$ and $\mM_j(\xi) = \mI_j(\xi)^{-1}\mB_j^{\top}\mF_j(\xi)^{p-1}\mB_j\mI_j(\xi)^{-1}$ with $\mB_j = \left.\frac{\partial \vf(\vbeta)}{\partial \vbeta^{\top}}\right|_{\vbeta = \vbeta_j}$ and $\mF_j(\xi) = \mB_j\mI_j(\xi)^{-1}\mB_j^{\top}$.
For convenience, we denote the right side of \eqref{eqn:suffweight1} as $d_p(\vx_i,\xi^{\vlambda^*})$.
For \emph{any} weight vector $\vlambda = [\lambda_1,\ldots,\lambda_n]^\top$, with simple linear algebra, it is easy to obtain
\begin{equation}\label{eqn:suffweight2}
\left\{\begin{array}{rlll}
q\sum\limits_{j=1}^m\tphi_0^j(\xi^{\vlambda})
 &=& \sum\limits_{i=1}^n\lambda_id_0(\vx_i,\xi^{\vlambda}), & p=0;\\
q^{1/p}\sum\limits_{j=1}^m\tphi_p^j(\xi^{\vlambda})\Phi_p^j(\xi^{\vlambda}) &=& \sum\limits_{i=1}^n \lambda_i d_p(\vx_i,\xi^{\vlambda}), & p>0.
\end{array}\right.
\end{equation}
Combining \eqref{eqn:suffweight1} and \eqref{eqn:suffweight2}, the sufficient condition of the optimal weights is equivalent to
\begin{equation}\label{eqn:suffweight3}
\sum\limits_{s=1}^n\lambda_s^*d_p(\vx_s,\xi^{\vlambda^*}) =  d_p(\vx_i,\xi^{\vlambda^*}), \,\,\,\, p\geq 0,
\end{equation}
for all design points $\vx_1,\cdots,\vx_n$. To obtain optimal weight $\vlambda^*$ that minimizes $\se(\xi^{\vlambda},\modelspace)$, the current weights of the design points could be adjusted according to the two sides of \eqref{eqn:suffweight3}.
For a design point $\vx_i$, if $d_p(\vx_i,\xi^{\vlambda})>\sum\limits_{s=1}^n\lambda_sd_p(\vx_s,\xi^{\vlambda})$, then the weight of point $\vx_i$ should be increased based on \eqref{eqn:suffweight3}. On the contrary, if $d_p(\vx_i,\xi^{\vlambda})<\sum\limits_{s=1}^n\lambda_sd_p(\vx_s,\xi^{\vlambda})$, the weight of point $\vx_i$ should be decreased based on \eqref{eqn:suffweight3}.
Thus, following the similar idea in classic multiplicative algorithms \citep{78silvey, 10yu},
the ratio $\left(d_p(\vx_i,\xi^{\vlambda})\left/\sum\limits_{s=1}^n\lambda_sd_p(\vx_s,\xi^{\vlambda})\right.\right)^{\delta}$ would be a good adjustment for the weight of design point $\vx_i$.
Since this weight updating scheme is inspired by the classic multiplicative algorithm, we call it a modified multiplicative procedure and describe it in Algorithm \ref{alg:weight} in Appendix.

We should remark that \cite{10yu} proved the convergence of classical multiplicative algorithm \citep{78silvey} to construct local optimal design for a class of optimality $\tr(\mI(\xi^{\vlambda}, M)^p), p<0$, and \cite{li2018efficient} extended the results to a more general class of $\Phi_p$-optimality.
However, the proof in \cite{10yu} can not be easily extended to prove the convergence of Algorithm \ref{alg:weight} since the derivative of $\se(\xi^{\vlambda}, \modelspace)$ to $\lambda_i$ cannot be reformulated into the general form in Equation (2) in \cite{10yu} where only one model is involved.
Nevertheless, Lemma \ref{lem:ConvWeights} has shown that the optimization problem solved by Algorithm \ref{alg:weight} is a convex optimization
\begin{equation}\label{eqn:optprob}
\begin{array}{rrclcl}
\displaystyle \min_{\vlambda} & \multicolumn{3}{l}{\se(\xi^{\vlambda},\modelspace) = \sum\limits_{j=1}^m \exp\left(\frac{\Phi_p^j(\xi^{\vlambda})}{\Phi_p^{\opt_j}}\right)} \\
\textrm{s.t.} & \mathbf{1}^{\top}\vlambda = 1, \ \vlambda \geq  \mathbf{0}
\end{array}
\end{equation}
with linear constraints.
Some existing optimization tools is available to solve such convex optimization.
Based on on our empirical study, Algorithm \ref{alg:weight} can converge to a solution as good as those from the commonly-used optimization tools, but has much faster computational speed.

 To show the strength of the proposed Algorithm \ref{alg:weight}, we use a small and simple example with three design points and two $\beta$ values. In the following \emph{Example 1}, we compare Algorithm \ref{alg:weight} with two existing convex optimization tools, \verb|fmincon| function in \textsc{Matlab} using interior-point method and the \verb|CVX| toolbox in \textsc{Matlab} for convex optimization.
To solve an optimization problem with the exponential objective function, \verb|CVX| constructed a successive approximation heuristic that approximates the local exponential function with polynomial approximation and solves the approximate model using symmetric primal/dual solvers \citep{grant2009cvx}.


\emph{Example 1.} Consider a univariate logistic regression model with the experimental domain $\Omega = [-1,1]$, basis function $\vg = [1,x]^{\top}$ and a parameter space $\mathcal{B} = \{\vbeta_1, \vbeta_2\}$ consisting of only two possible regression coefficients $\vbeta_1 = [-1.4,2.3]^{\top}$ and $\vbeta_2 = [0.5,1.2]^{\top}$.
The model space is $\mathcal{M} = \{M_1 = (h,\vg,\vbeta_1),M_2 = (h,\vg,\vbeta_2)\}$, where $h$ is the link function of logistic regression.
Given design points $x \in \{-1,0,1\}$, all three optimization methods return the same optimal weights,
\[\vlambda^* = \{0.3832,0.2660,0.3508\}.\]
Table \ref{tab:comptime} reports the computational times of the three comparison methods.
The results clearly show that Algorithm \ref{alg:weight} is far more efficient than both \verb|CVX| and \verb|fmincon|.
Furthermore, Algorithm \ref{alg:weight} boosts the speed of sequential Algorithm \ref{alg:sequential} dramatically as finding the optimal weights is done in every iteration of the sequential algorithm.

\begin{table}[ht]
\centering
\caption{Computational Times (in seconds) of Three Optimization Methods. \label{tab:comptime}}
\begin{tabular}{|c|c|c|}\hline
CVX   & fmincon & Algorithm \ref{alg:weight} (Optimal-Weight Procedure)\\\hline
4.04  & 1.44  & 0.17 \\\hline
\end{tabular}%
\end{table}

It is worth pointing out that, occasionally, $\tphi_p^j(\xi^{(r)}) = \left[\Phi_p^{\opt_j}\right]^{-1}\exp\left(\frac{\Phi_p^j(\xi^{(r)})}{\Phi_p^{\opt_j}}\right)$ and $\tphi_p^j(\xi^{\vlambda^{(k)}}) = \left[\Phi_p^{\opt_j}\right]^{-1}\exp\left(\frac{\Phi_p^j(\xi^{\vlambda^{(k)}})}{\Phi_p^{\opt_j}}\right)$ in \eqref{for:multialg} of Algorithm \ref{alg:weight} and directional derivative $\phi(\vx,\xi^{(r)})$ in Algorithm \ref{alg:sequential} can get extreme large and cause overflow, which is a well-recognized issue with the Log-Sum-Exp approximation in the literature.
One remedy is to introduce a constant $c$, and $\exp\left(\frac{\Phi_p^j(\xi^{(r)})}{\Phi_p^{\opt_j}}\right) = 、\exp(c)\exp\left(\frac{\Phi_p^j(\xi^{(r)})}{\Phi_p^{\opt_j}}-c\right)$.
This constant scaling factor $\exp(c)$ is eventually canceled in \eqref{for:multialg} in Algorithm \ref{alg:weight} and does not affect the search for the next design point Algorithm \eqref{alg:sequential}.
We set $c = \left\lceil\max\limits_j\left(\frac{\Phi_p^j(\xi^{\vlambda^{(k)}})}{\Phi_p^{\opt_j}}\right)-500\right\rceil$ in Algorithm \ref{alg:weight} and $c = \left\lceil\max\limits_j\left(\frac{\Phi_p^j(\xi^{(r)})}{\Phi_p^{\opt_j}}\right)-500\right\rceil$ in Algorithm \ref{alg:sequential} whenever overflow occurs.

\section{Numerical Examples}\label{sec:examples}

In this section, we conduct several numerical examples to evaluate the performance of the proposed Mm-$\Phi_p$ design under different types of model uncertainty.
The performance of the proposed Mm-$\Phi_p$ design is compared to the compromise design proposed by \cite{woods2006designs}.
As we have clarified in Section \ref{subsec:robvscom}, there are two types of compromise design.
The eff-compromise design $\effcomdes$ aims at maximizing the average $\Phi_p$-efficiency and the $\Phi_p$-compromise design $\phicomdes$ aims at minimizing the average $\Phi_p$-optimality criterion.
The later one coincides with the Bayesian optimal design when only considering the uncertainty from unknown regression coefficients.
For all the designs in the examples, the candidate pool $\mathcal{C}$ is constructed by grid points and each dimension of $\vx$ has 51 equally spaced grid points.
We use the default uniform prior distribution on the model specification for the compromise designs. For $\vf(\vbeta) = [f_1(\vbeta),...,f_q(\vbeta)]^{\top}$ in $\Phi_p(\xi,M)$ in \eqref{eq:phi_p}, we set $f_{j}(\vbeta) = \beta_{j}$.

\subsection{Model Uncertainty}
In the following \emph{Example 2}, we investigate the performance of the Mm-$\Phi_p$ design and algorithm when the uncertainties are involved in both the link functions and basis functions in the model space $\mathcal{M}$.

\emph{Example 2.} For an experiment with $d=2$ input variables and one binary response, consider both logistic regression model and probit model, and possible polynomial basis functions up to degree 2, i.e.,
$$\mathcal{G} = \left\{\vg_1 = (1,x_1,x_2)^{\top},
\vg_2 = (1, x_1,x_2,x_1x_2)^{\top}, \vg_3 = (1, x_1,x_2,x_1x_2,x_1^2,x_2^2)^{\top} \right\}.$$
For the basis $\vg_3$, the regression coefficients $\vbeta_3 = [\beta_{3,1},\cdots,\beta_{3,6}]^{\top}$ are drawn randomly from standard multivariate normal distribution.
For the basis $\vg_2$, the the regression coefficients $\vbeta_2 = [\beta_{2,1},\cdots, \beta_{2,4}]^{\top}$ are drawn independently with $\beta_{2,j} \sim \text{N}(\beta_{3,j}, (0.5\beta_{3,j})^2)$, for $j = 1,2,3,4$.
The variance $(0.5\beta_{3,j})^2$ that depends on the regression coefficient $\beta_{3,j}$ allows a larger perturbation for $\beta_{2,j}$ when the corresponding $\beta_{3,j}$ is large.
It is to accommodate the fact that the values of the regression coefficients are likely to change when the quadratic terms are removed.
For the basis $\vg_1$, the regression coefficients $\vbeta_1 = [\beta_{1,1},\beta_{1,2},\beta_{1,3}]^{\top}$ are drawn independently with $\beta_{1,i} \sim \text{N}(\beta_{3,i}, (0.5\beta_{3,i})^2)$, for $i=1,2,3$. Thus, the model space $\mathcal{M}$ consists of six models: $\mathcal{M} = \{M_1 = (\text{probit}, \vg_1,\vbeta_1), M_2 = (\text{probit}, \vg_2,\vbeta_2), M_3 = (\text{probit}, \vg_3,\vbeta_3), M_4 = (\text{logit}, \vg_1,\vbeta_1), M_5 = (\text{logit}, \vg_2,\vbeta_2), M_6 = (\text{logit}, \vg_3,\vbeta_3)\}$.
We generate 100 parameter sets $\mathcal{B} = \{\vbeta_1, \vbeta_2, \vbeta_3\}$ to form 100 model sets.
For each generated model set, the Mm-$\Phi_p$ design, eff-compromise design, and $\Phi_p$-compromise design are constructed, respectively.

To compare the designs, we use the $\Phi_p$-efficiency defined in \eqref{eqn:phieff} as a larger-the-better performance measure.
In particular, we consider $\Phi_0(\xi,M)$ (i.e., $\lim\limits_{p\rightarrow 0}\Phi_p(\xi,M)$) which is the D-optimality and $\Phi_1(\xi,M)$ which is the A-optimality.
For each model space, we compute the $\Phi_p$-efficiency in \eqref{eqn:phieff} of all three designs relative to the corresponding local optimal design, and the local optimal design $\optdesM$ is obtained by the algorithm of \cite{li2018efficient}.
For each model space, we can calculate the worse-case efficiency as $\min\limits_{M_i\in\mathcal{M}}\eff_{\Phi_p}(\xi, \xi^{\opt}_{M_i}; M_i)$.

\begin{figure}[hbtp]
\centering
\subfloat[A-optimality]
{{\includegraphics[width=6cm]{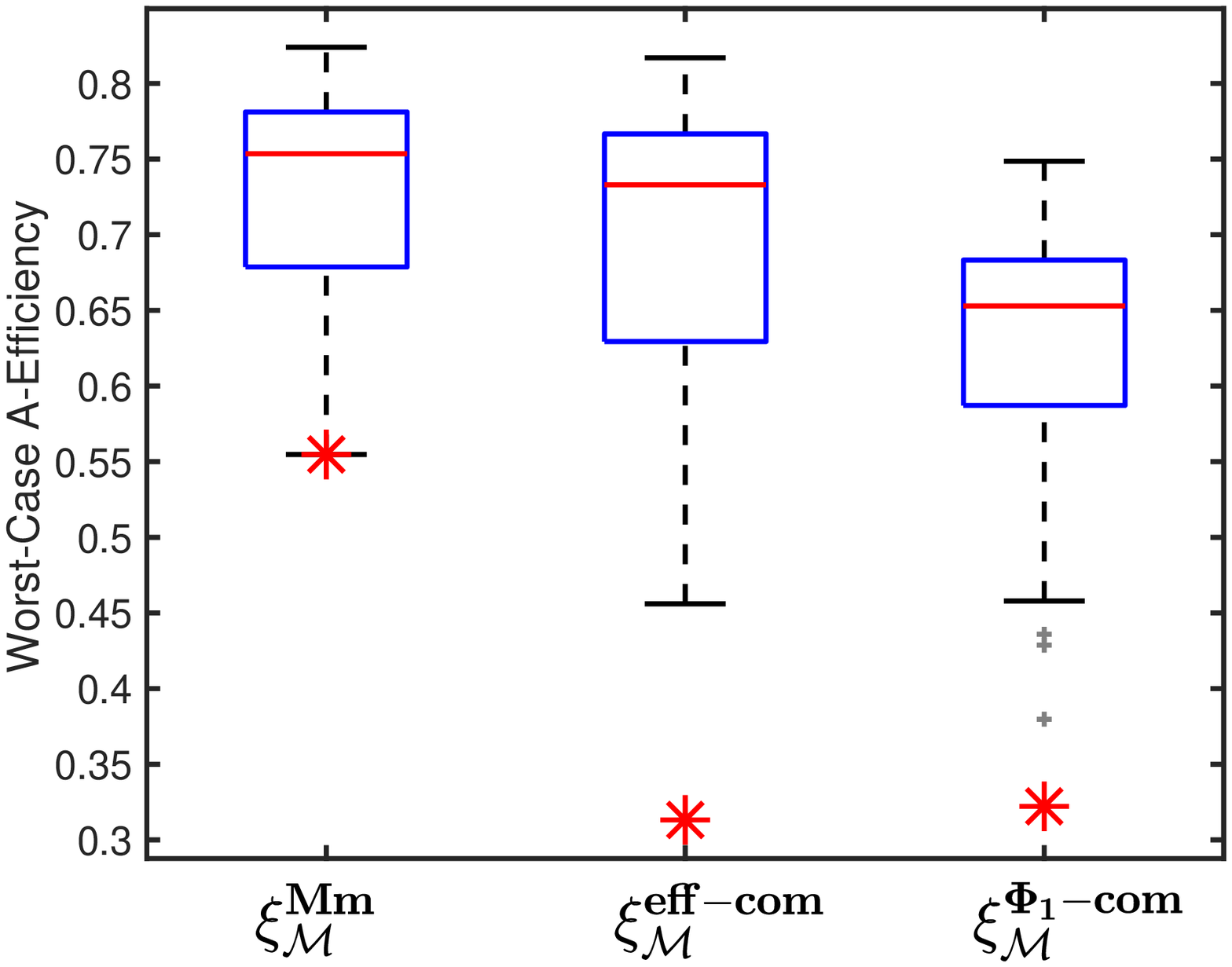}}}
\qquad
\subfloat[D-optimality]
{{\includegraphics[width=6cm]{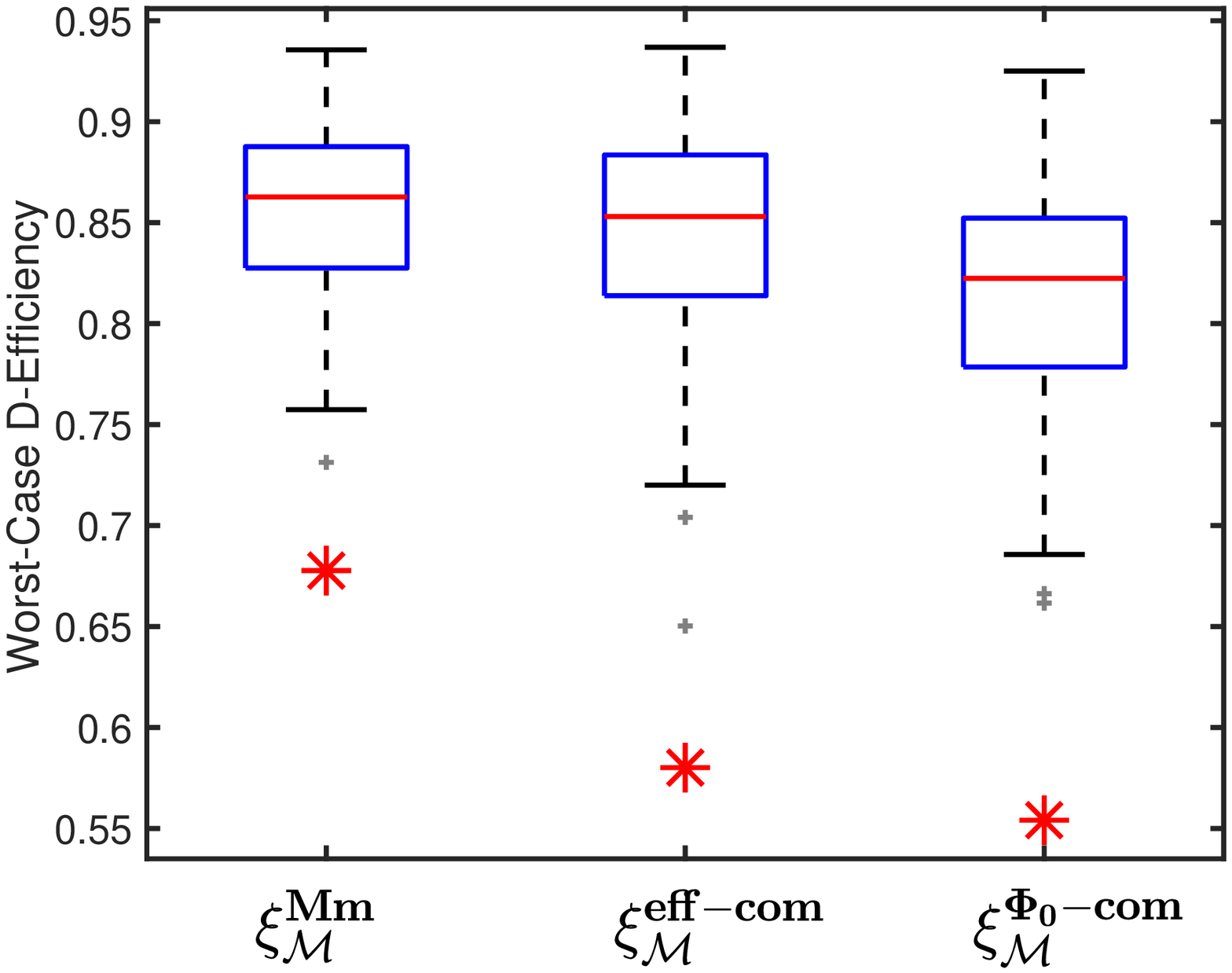}}}
\caption{Boxplot of Worse-Case A- and D-Efficiency of Mm-$\Phi_p$ Design, Eff-Compromise Design, and $\Phi_p$-Compromise Design across 100 Randomly Generated Model Spaces}
\label{fig:2dProLogBoxPlot}
\end{figure}

\begin{table}[htbp]
  \centering
  \caption{Minimum and Median of the Worst-Case A- and D-Efficiency across 100 Randomly Generated Model Spaces for Comparison of Designs}
   \begin{tabular}{@{\extracolsep{4pt}}|lcccc|@{}}
    \hline
    & \multicolumn{2}{c}{Worst-Case A-Efficiency} & \multicolumn{2}{c}{Worst-Case D-Efficiency}\vline\\
 \cline{2-3}  \cline{4-5}
          &   $\min$    & $\median$  & $\min$ & $\median$ \\
    \hline
    Mm-$\Phi_p$ Design & 0.55 & 0.75 &0.68 & 0.86\\
    Eff-Compromise Design & 0.31  & 0.73 &0.58 & 0.85\\
    $\Phi_p$-Compromise Design & 0.32  & 0.65 &0.55 & 0.82\\
    \hline
    \end{tabular}%
  \label{tab:2dProLogADEff}%
\end{table}

Figure \ref{fig:2dProLogBoxPlot} shows the boxplot of the worst-case A- and D-efficiency of the Mm-$\Phi_p$ design, eff-compromise design and $\Phi_p$-compromise design across 100 different model sets. The red asterisks ``$*$" in the boxplot denote the minimum worst-case A- and D-efficiency, and the larger the minimum, the better the design. Table \ref{tab:2dProLogADEff} summarizes the minimum and median of the worst-case A- and D-efficiency of the three designs.
The results show that the Mm-$\Phi_p$ design returns the largest values on the minimum and median of the worst-case efficiency.
We also notice that the eff-compromise design often gives the highest mean efficiency for a given model space, which is expected since it is designated to achieve the maximum mean efficiency.
However, the mean A- and D-efficiency of all three designs are comparable on average over the 100 model sets.
The computational times of Algorithm 1 to construct the Mm-$\Phi_p$ design are about 7.59 seconds and 6.18 seconds for A- and D-optimality, respectively.

\subsection{Uncertain Regression Coefficients}
In the following \emph{Example 3}, we further illustrate the advantages of Mm-$\Phi_p$ design through an example considering the uncertain regression coefficients with the specified link function $h$ and basis functions $\vg$.
Note that when the regression coefficient space $\mathcal{B}$ is continuous, a discretization is needed.
In \emph{Example 3}, the performance of the proposed design and algorithm over the unsampled values of regression coefficient $\vbeta$ is investigated.

\emph{Example 3}.  For a univariate logistic regression model with experimental domain $\Omega = [-1,1]$ and  a quadratic basis, i.e. $\vg (x) = [1,x,x^2]^{\top}$, consider a regression coefficient space $\mathcal{B} = \{\beta_1\in[0,6],\beta_2\in[-6,0],\beta_3\in[5,11]\}$.
Since $\mathcal{B}$ is continuous, we choose a Sobol sample of size twenty-six and the centroid $\vbeta_c = [3,-3,8]^{\top}$ of $\mathcal{B}$, i.e. $m=27$,  to form the surrogate coefficient set $ \mathcal{B}'$.
Sobol sample is a low discrepancy sequence that converges to a uniform distribution on a bounded set and it is widely used in Monte Carlo methods \citep{sobol1967distribution}.
The surrogate model set is $\mathcal{M}' = \{M = (h,\vg,\vbeta): h, \vg, \vbeta\in \mathcal{B}'\}$, where $h$ is the link function of the logistic regression.
Four designs are considered: (1) Mm-$\Phi_p$ design $\robdes$; (2) eff-compromise design $\effcomdes$; (3) local optimal design $\xi^{\text{center}}$ of the centroid of $\mathcal{B}$, i.e. $\vbeta_c = [3,-3,8]^{\top}$, which can be viewed as either Mm-$\Phi_p$ or compromise design with $m=1$, and (4) Bayesian optimal design $\xi^{\text{Bayesian}}_{\modelspace}$ with uniform prior, which is also the $\Phi_p$-compromise design.
Figure \ref{fig:1dLogitPoints} shows the constructed designs under D- and A-optimality, respectively.

To compare the four designs, we use the $\Phi_p$-efficiency defined in \eqref{eqn:phieff} as a performance measure.
Specifically, we generate a size of 10,000 Sobol sample from the original continuous region $\mathcal{B}$.
For each of the sample, we compute the $\Phi_p$-efficiency in \eqref{eqn:phieff} of all four designs relative to the corresponding local optimal design,
and the local optimal design $\optdesM$ is obtained in the same way as in Example 1.
Figure \ref{fig:1dLogitBoxPlot} shows the boxplot of A- and D-efficiency of $\robdes$, $\effcomdes$, $\xi^{\text{center}}$, and $\xi^{\text{Bayesian}}_{\modelspace}$ over 10,000 randomly sampled $\vbeta$ values. The red asterisks ``$*$" in the boxplot denote the worst-case A- and D-efficiency, and the larger the worst-case efficiency, the better the design. Table \ref{tab:1dLogitIDEff} summarizes the minimum and median A- and D-efficiency of the four designs.

It is seen that the Mm-$\Phi_p$ design $\robdes$ outperforms the other three designs in terms of the worst-case design efficiency, especially for A-optimality.
Specifically, the worst-case A-efficiency of the Mm-$\Phi_p$ design is 0.41, and is much larger than those of the other three designs.
The worst-case D-efficiency of the Mm-$\Phi_p$ design is 0.86 and is only slightly larger than those of other designs.
We also found that the maximum A-efficiency of the Mm-$\Phi_p$ design is the smallest,
which is not surprising considering the Mm-$\Phi_p$ design maximizes the worst-case efficiency, not the best-case efficiency.

To illustrate the computational efficiency of the proposed Algorithm \ref{alg:sequential}, Figure \ref{fig:1dLogitIterVsObg} shows how the Mm-$\Phi_p$ design criterion $\lse(\xi^{(r)},\mathcal{M}')$ decreases with respect to the number of iterations.
The computation times of Algorithm 1 to construct $\robdes$ are 1.68 and 1.47 seconds for A- and D-optimality, respectively.

\begin{figure}[hbtp]
\centering
\subfloat[A-optimality]
{{\includegraphics[width=6cm]{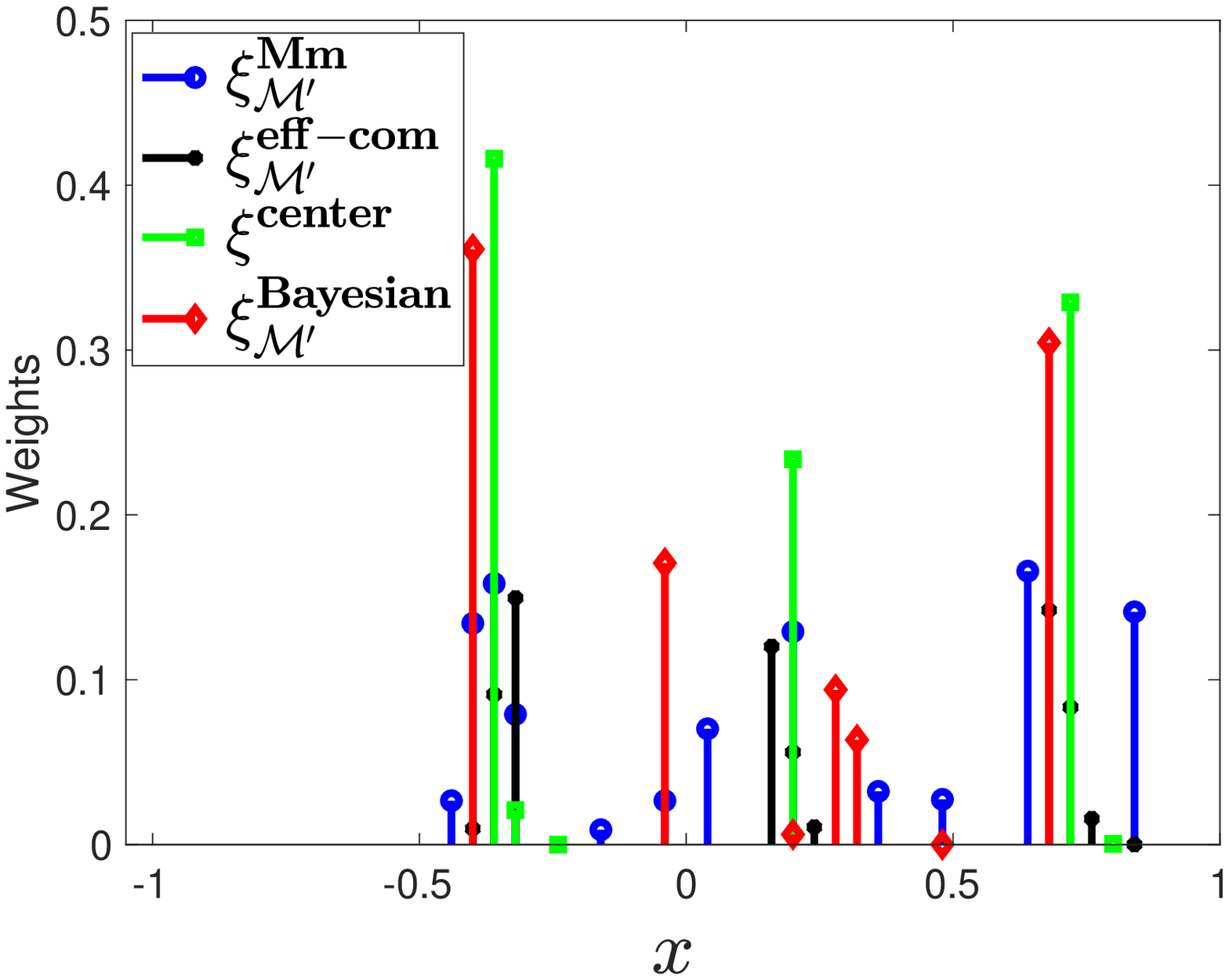}}}
\qquad
\subfloat[D-optimality]
{{\includegraphics[width=6cm]{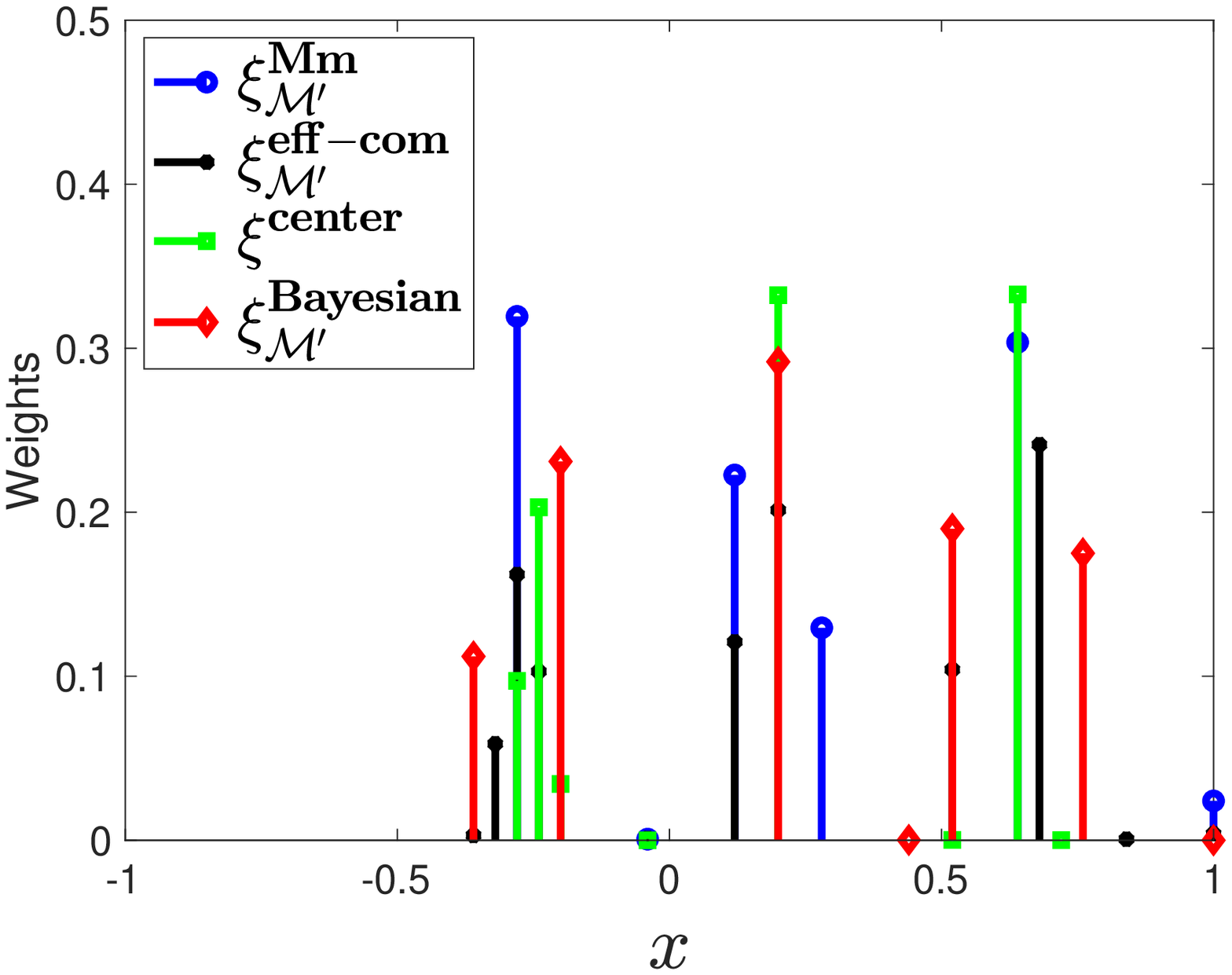}}}
\caption{Mm-$\Phi_p$ Design, Eff-Compromise Design, Centroid Optimal Design and Bayesian Optimal Design}
\label{fig:1dLogitPoints}
\end{figure}

\begin{table}[htbp]
  \centering
  \caption{Minimum and Median of A- and D- Efficiency across 10,000 Sampled $\vbeta$ for Comparison of Four Designs}
    \begin{tabular}{@{\extracolsep{4pt}}|lcccc|@{}}
    \hline
    & \multicolumn{2}{c}{A-Efficiency} & \multicolumn{2}{c}{D-Efficiency}\vline\\
 \cline{2-3}  \cline{4-5}
          &   $\min$    & $\median$  & $\min$ & $\median$ \\
    \hline
    Mm-$\Phi_p$ Design & 0.41  & 0.70 &0.86 & 0.98\\
    Eff-Compromise Design & 0.21  & 0.71 &0.83 & 0.98\\
    Centroid Optimal Design & 0.16  & 0.71 & 0.81& 0.98\\
    Bayesian Optimal Design & 0.26  & 0.69 &0.84 & 0.98\\
    \hline
    \end{tabular}%
  \label{tab:1dLogitIDEff}%
\end{table}

\begin{figure}[hbtp]
\centering
\subfloat[A-optimality]
{{\includegraphics[width=6cm]{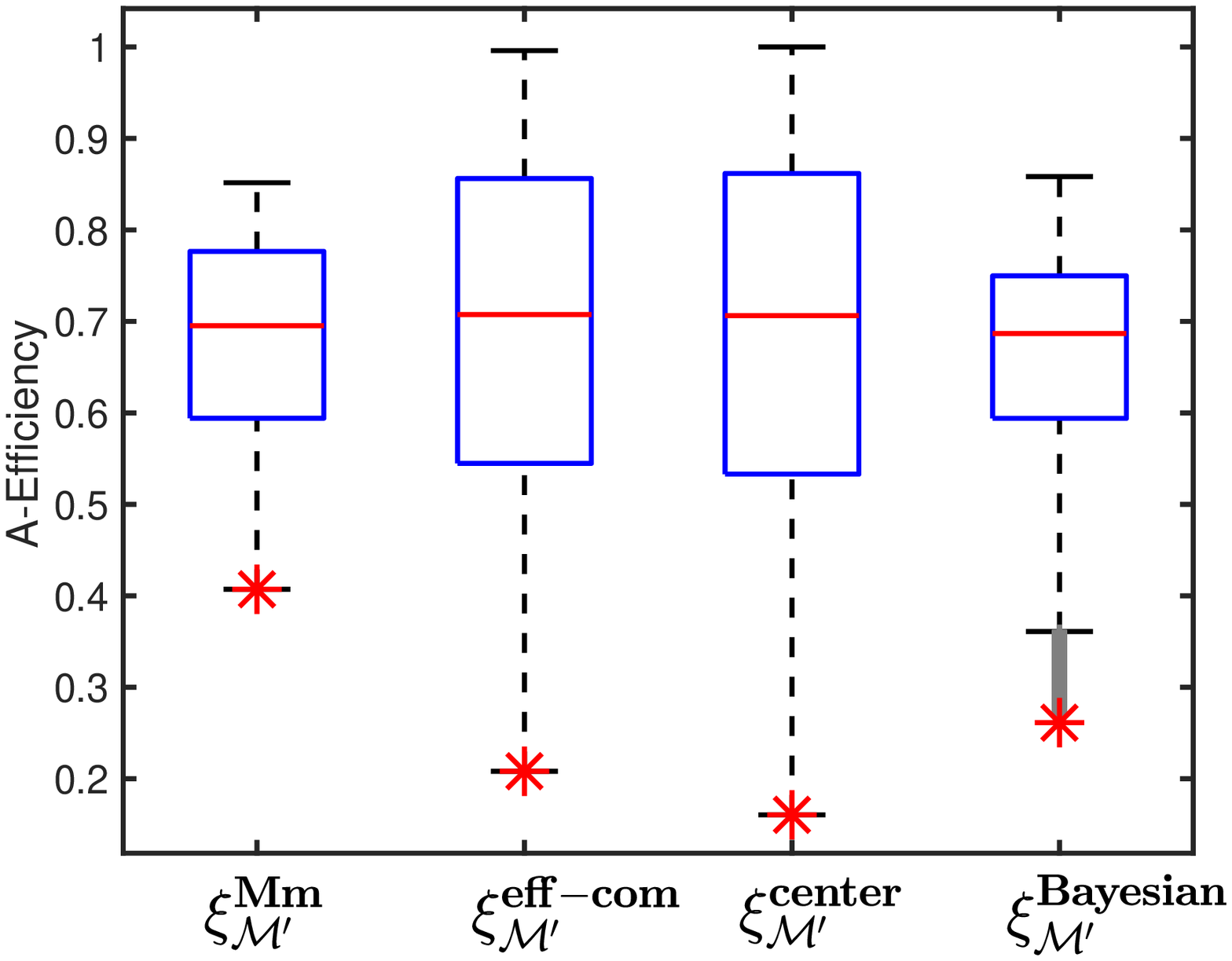}}}
\qquad
\subfloat[D-optimality]
{{\includegraphics[width=6cm]{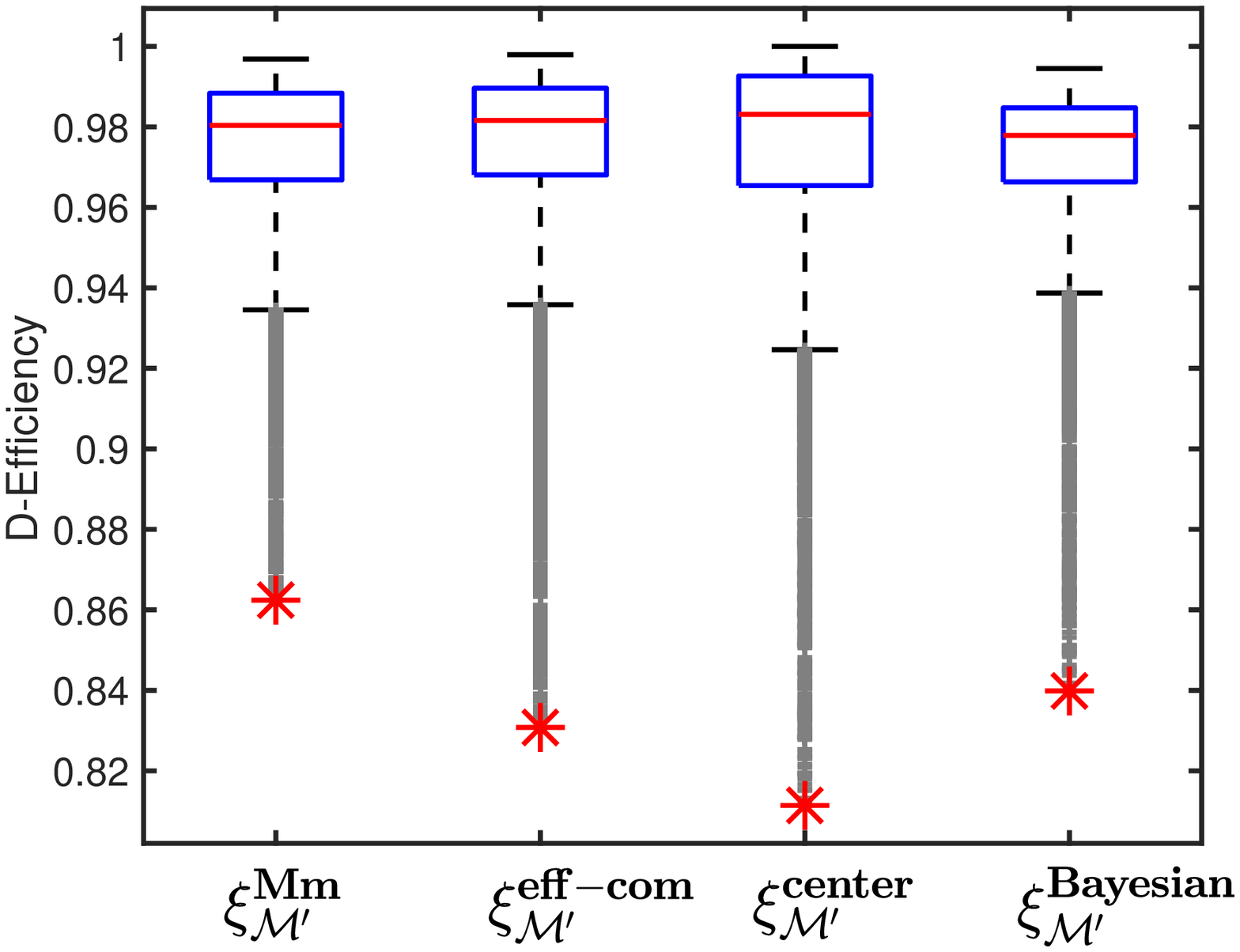}}}
\caption{Boxplot of A- and D-Efficiency of Four Designs at 10,000 Sampled $\vbeta$}
\label{fig:1dLogitBoxPlot}
\end{figure}

\begin{figure}[hbtp]
\centering
\subfloat[A-optimality]
{{\includegraphics[width=6cm]{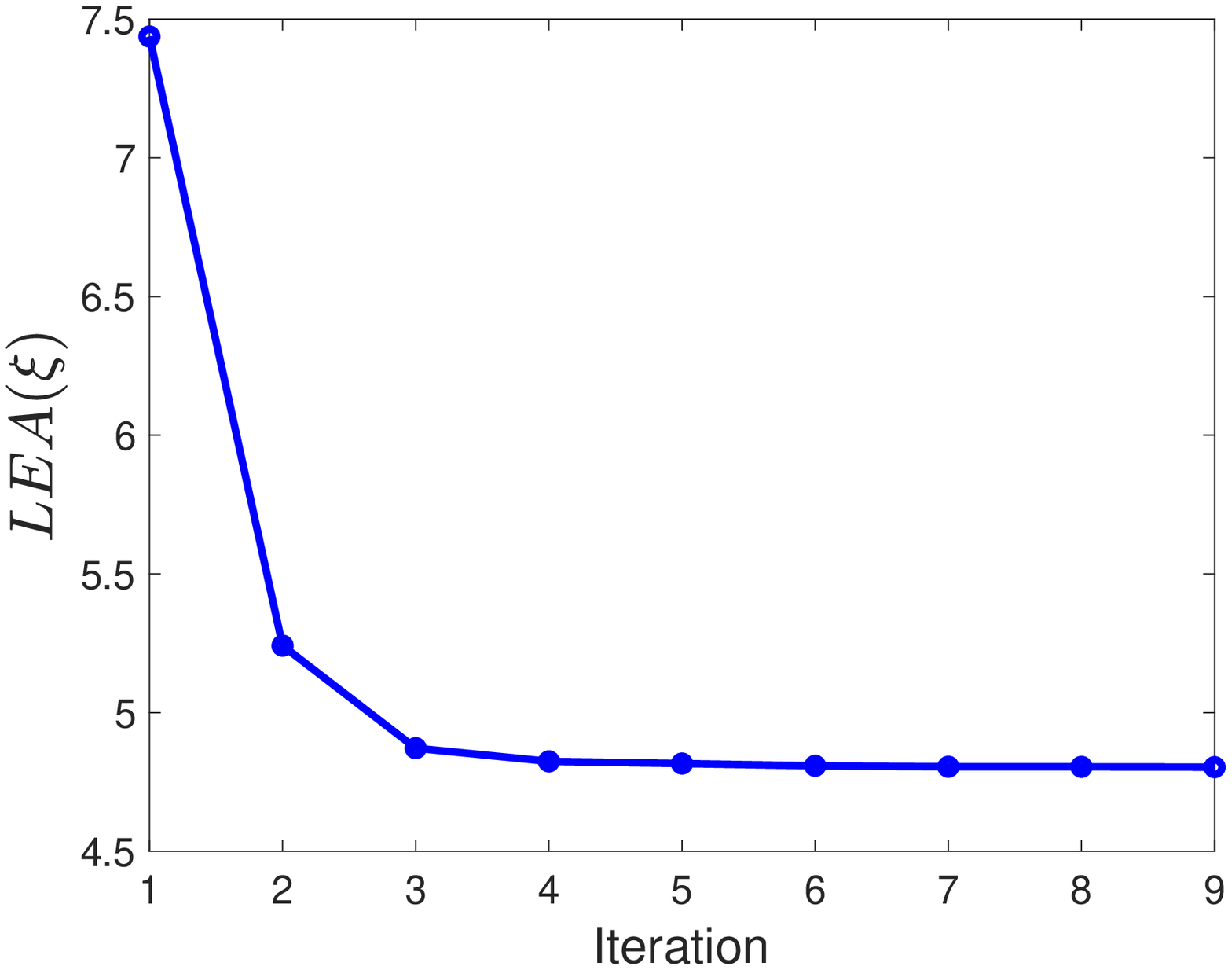}}}
\qquad
\subfloat[D-optimality]
{{\includegraphics[width=6cm]{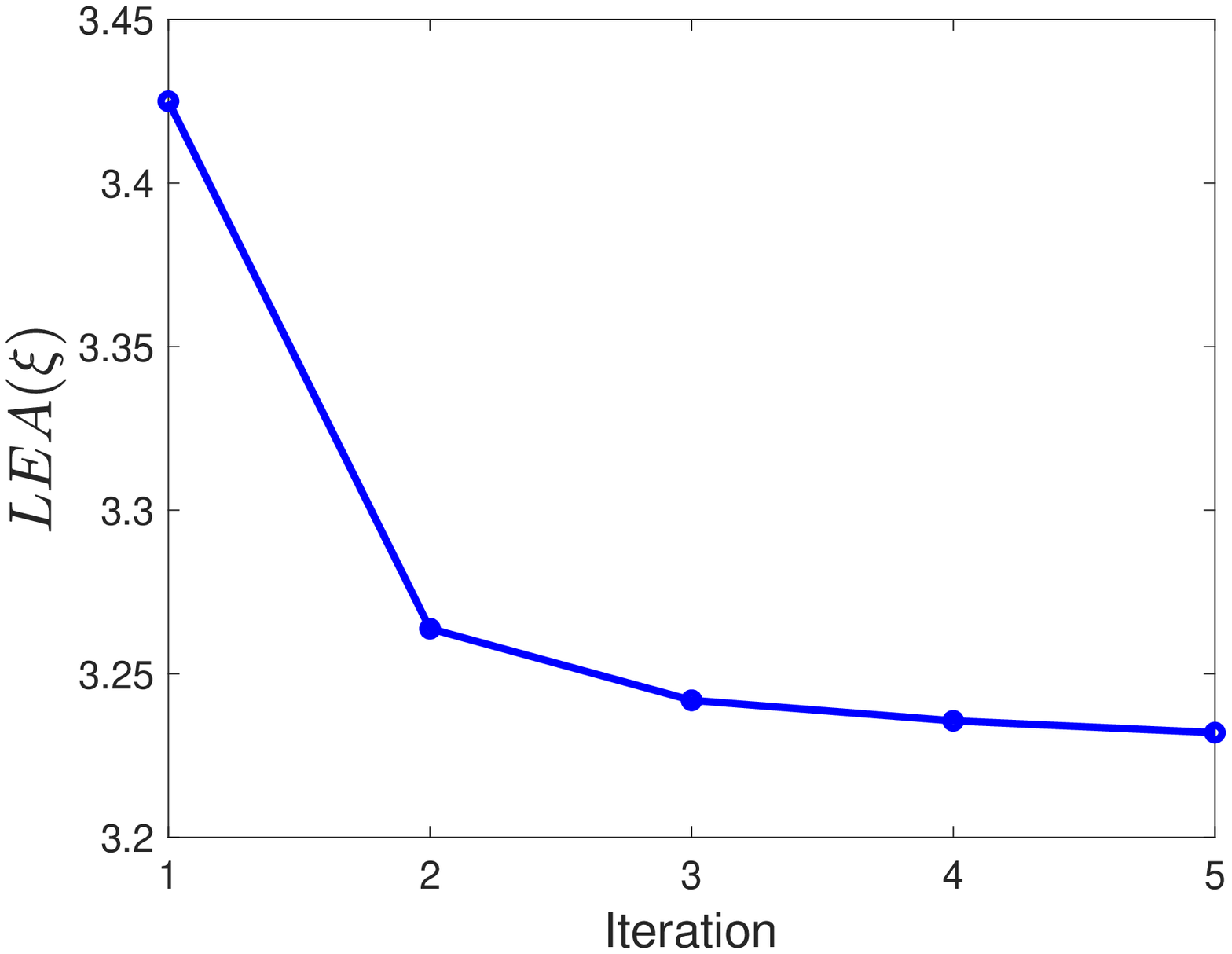}}}
\caption{$\lse(\xi^{(r)},\mathcal{M}')$ of the $r$-th iteration in Algorithm \ref{alg:sequential}.}
\label{fig:1dLogitIterVsObg}
\end{figure}

\subsection{Potato Packing Example}\label{sec: potato}

We consider a real-world example, the potato packing example in \cite{woods2006designs}, to further evaluate the proposed Mm-$\Phi_p$ design.
The experiment contains $d=3$ quantitative variables - vitamin concentration in the prepackaging dip and the amount of two kinds of gas in the packing atmosphere.
The response is binary representing the presence or absence of liquid in the pack after 7 days.
The basis functions of the logistic regression model always include the linear and quadratic terms of the input variables.
But one set of the basis functions contains the interaction terms and the other one does not.
The estimates of regression coefficients from a preliminary study in \cite{woods2006designs} are given in Table \ref{tab:PotatoPackModel} in the Appendix.
Since enhancing prediction accuracy is a major goal for the experiment, we use the prediction-oriented I-optimality \citep{atkinson2014optimal} to evaluate the design efficiency.
Note that the I-optimality shares the same mathematical structure as $\Phi_1$-optimality.
The design points of the designs are shown in Figure \ref{fig:PotatoPackPoints} in the Appendix.
Table \ref{tab:PotatoPackIEff} summarizes the I-efficiency of the Mm-$\Phi_p$ design, eff-compromise design, and I-compromise design of the three potential model specifications.
In terms of worst-case efficiency (i.e., smallest value of I-efficiency among $M_1$, $M_{2}$ and $M_{3}$), the proposed Mm-$\Phi_p$ design outperforms the other two designs by a large margin.

\begin{table}[htbp]
  \centering
  \caption{I-Efficiency of Mm-$\Phi_p$ Design, Eff-Compromise Design and I-Compromise Design}
    \begin{tabular}{|lrrr|}
    \hline
          &   $M_1$   & $M_2$ & $M_3$ \\
    \hline
    Mm-$\Phi_p$ Design & 0.64  & 0.71 & 0.82  \\
    Eff-Compromise Design & 0.52  & 0.78 & 0.92 \\
    I-Compromise Design & 0.49  & 0.80 & 0.92 \\
    \hline
    \end{tabular}%
  \label{tab:PotatoPackIEff}%
\end{table}

%

\section{Discussion}\label{sec:discussion}

In this article, we proposed a new maximin $\Phi_p$-efficiency criterion Mm-$\Phi_p$ for GLMs that aims at maximizing the worst-case design efficiency when various kinds of model uncertainties are considered, including uncertainties in link function, linear predictor and regression coefficients.
An efficient algorithm to construct the Mm-$\Phi_p$ design is also developed based on sound theoretical properties of the criterion.
The proposed Mm-$\Phi_p$ design and the algorithm can be easily extended to a more general case such as nonlinear models in \cite{13yang}. 

There are several directions for further research to enhance the proposed Mm-$\Phi_p$ design and algorithm.
First, to construct the Mm-$\Phi_p$ design, one needs to form a set of possible model specifications.
An interesting direction is how to extend the proposed design when such information of model specifications is limited or unavailable.
Second, it is interesting to rigorously establish the convergence property of the optimal-weight procedure  (Algorithm \ref{alg:weight}),
which requires developing some other mathematical results.
Third, the use of log-sum-exp approximation can be applied to other maximin design with a convex design criterion, and the theoretical and algorithmic developments can be adapted similarly.
We plan to extend the framework to the more general setting for other maximin designs with convex criterion.

\bibliography{GLM}

\begin{thebibliography}{39}
\newcommand{\enquote}[1]{``#1''}
\expandafter\ifx\csname natexlab\endcsname\relax\def\natexlab#1{#1}\fi

\bibitem[{Amzal et~al.(2006)Amzal, Bois, Parent, and
  Robert}]{amzal2006bayesian}
Amzal, B., Bois, F.~Y., Parent, E., and Robert, C.~P. (2006),
  \enquote{Bayesian-Optimal Design via Interacting Particle Systems,}
  \textit{Journal of the American Statistical Association}, 101, 773--785.

\bibitem[{Atkinson(2014)}]{atkinson2014optimal}
Atkinson, A.~C. (2014), \enquote{Optimal Design,} \textit{Wiley StatsRef:
  Statistics Reference Online}, 1--17.

\bibitem[{Atkinson et~al.(2006)Atkinson, Donev, and Tobias}]{06atk}
Atkinson, A.~C., Donev, A.~N., and Tobias, R.~D. (2006), \textit{Optimum
  Experimental Designs, with SAS}, Oxford University Press.

\bibitem[{Atkinson and Woods(2015)}]{atkinson2015designs}
Atkinson, A.~C. and Woods, D.~C. (2015), \enquote{Designs for generalized
  linear models,} \textit{Handbook of Design and Analysis of Experiments},
  471--514.

\bibitem[{Burghaus and Dette(2014)}]{burghaus2014optimal}
Burghaus, I. and Dette, H. (2014), \enquote{Optimal designs for nonlinear
  regression models with respect to non-informative priors,} \textit{Journal of
  Statistical Planning and Inference}, 154, 12--25.

\bibitem[{Calafiore and El~Ghaoui(2014)}]{calafiore2014optimization}
Calafiore, G.~C. and El~Ghaoui, L. (2014), \textit{Optimization models},
  Cambridge university press.

\bibitem[{Dean et~al.(2015)Dean, Morris, Stufken, and
  Bingham}]{dean2015handbook}
Dean, A., Morris, M., Stufken, J., and Bingham, D. (2015), \textit{Handbook of
  design and analysis of experiments}, vol.~7, CRC Press.

\bibitem[{Dror and Steinberg(2006)}]{dror2006robust}
Dror, H.~A. and Steinberg, D.~M. (2006), \enquote{Robust experimental design
  for multivariate generalized linear models,} \textit{Technometrics}, 48,
  520--529.

\bibitem[{Fedorov and Leonov(2013)}]{fedorov2013optimal}
Fedorov, V. and Leonov, S. (2013), \textit{Optimal Design for Nonlinear
  Response Models}, Chapman \& Hall/CRC Biostatistics Series.

\bibitem[{Fellman(1974)}]{74fellman}
Fellman, J. (1974), \textit{On the Allocation of Linear Observations}, Societas
  Scientiarum Fennica.

\bibitem[{Fiacco and Kortanek(1983)}]{83tor}
Fiacco, A.~V. and Kortanek, K.~O. (eds.) (1983), \textit{A Moment Inequality
  and Monotonicity of an Algorithm}.

\bibitem[{Grant et~al.(2009)Grant, Boyd, and Ye}]{grant2009cvx}
Grant, M., Boyd, S., and Ye, Y. (2009), \enquote{cvx users? guide,}
  \textit{online: http://www. stanford. edu/\~{} boyd/software. html}.

\bibitem[{Imhof and Wong(2000)}]{imhof2000graphical}
Imhof, L. and Wong, W.~K. (2000), \enquote{A graphical method for finding
  maximin efficiency designs,} \textit{Biometrics}, 56, 113--117.

\bibitem[{Khuri et~al.(2006)Khuri, Mukherjee, Sinha, and
  Ghosh}]{khuri2006design}
Khuri, A.~I., Mukherjee, B., Sinha, B.~K., and Ghosh, M. (2006),
  \enquote{Design issues for generalized linear models: A review,}
  \textit{Statistical Science}, 376--399.

\bibitem[{Kiefer(1974)}]{kiefer1974general}
Kiefer, J. (1974), \enquote{General equivalence theory for optimum designs
  (approximate theory),} \textit{The annals of Statistics}, 849--879.

\bibitem[{Kiefer(1985)}]{75kiefer}
--- (1985), \enquote{Optimal Design: Variation in Structure and Performance
  Under Change of Criterion,} \textit{Biometrika}, 62, 277--288.

\bibitem[{King and Wong(1998)}]{king1998optimal}
King, J. and Wong, W.~K. (1998), \enquote{Optimal minimax designs for
  prediction in heteroscedastic models,} \textit{Journal of Statistical
  Planning and Inference}, 69, 371--383.

\bibitem[{King and Wong(2000)}]{king2000minimax}
King, J. and Wong, W.-K. (2000), \enquote{Minimax D-optimal designs for the
  logistic model,} \textit{Biometrics}, 56, 1263--1267.

\bibitem[{Li and Majumdar(2009)}]{li2009some}
Li, G. and Majumdar, D. (2009), \enquote{Some results on D-optimal designs for
  nonlinear models with applications,} \textit{Biometrika}, 96, 487--493.

\bibitem[{Li and Deng(2018)}]{li2018efficient}
Li, Y. and Deng, X. (2018), \enquote{An Efficient Algorithm for Elastic
  I-optimal Design of Generalized Linear Models,} \textit{arXiv preprint
  arXiv:1801.05861}.

\bibitem[{Schwabe(1997)}]{schwabe1997maximin}
Schwabe, R. (1997), \enquote{Maximin efficient designs Another view at
  D-optimality,} \textit{Statistics \& probability letters}, 35, 109--114.

\bibitem[{Silvey(2013)}]{silvey2013optimal}
Silvey, S.~D. (2013), \textit{Optimal design: an introduction to the theory for
  parameter estimation}, vol.~1, Springer Science \& Business Media.

\bibitem[{Silvey et~al.(1978)Silvey, Titterington, and Torsney}]{78silvey}
Silvey, S.~D., Titterington, D.~M., and Torsney, B. (1978), \enquote{An
  algorithm for optimal designs on a finite design space,}
  \textit{Communications in Statistics - Theory and Methods}, 14, 1379--1389.

\bibitem[{Sitter(1992)}]{sitter1992robust}
Sitter, R.~R. (1992), \enquote{Robust designs for binary data,}
  \textit{Biometrics}, 1145--1155.

\bibitem[{Sitter and Torsney(1995)}]{sitter1995d-optimal}
Sitter, R.~R. and Torsney, B. (1995), \enquote{D-Optimal Designs for
  Generalized Linear Models,} in \textit{MODA4 --- Advances in Model-Oriented
  Data Analysis}, eds. Kitsos, C.~P. and M{\"u}ller, W.~G., Heidelberg:
  Physica-Verlag HD, pp. 87--102.

\bibitem[{Sobol'(1967)}]{sobol1967distribution}
Sobol', I.~M. (1967), \enquote{On the distribution of points in a cube and the
  approximate evaluation of integrals,} \textit{Zhurnal Vychislitel'noi
  Matematiki i Matematicheskoi Fiziki}, 7, 784--802.

\bibitem[{Waite and Woods(2015)}]{waite2015designs}
Waite, T.~W. and Woods, D.~C. (2015), \enquote{{Designs for generalized linear
  models with random block effects via information matrix approximations},}
  \textit{Biometrika}, 102, 677--693.

\bibitem[{Wong(1992)}]{wong1992unified}
Wong, W.-K. (1992), \enquote{A unified approach to the construction of minimax
  designs,} \textit{Biometrika}, 79, 611--619.

\bibitem[{Wong and Cook(1993)}]{wong1993heteroscedastic}
Wong, W.~K. and Cook, R.~D. (1993), \enquote{Heteroscedastic G-optimal
  designs,} \textit{Journal of the Royal Statistical Society: Series B
  (Methodological)}, 55, 871--880.

\bibitem[{Wong et~al.(2019)Wong, Yin, and Zhou}]{wong2019optimal}
Wong, W.~K., Yin, Y., and Zhou, J. (2019), \enquote{Optimal Designs for
  Multi-Response Nonlinear Regression Models With Several Factors via
  Semidefinite Programming,} \textit{Journal of Computational and Graphical
  Statistics}, 28, 61--73.

\bibitem[{Woods and Lewis(2011)}]{woods2011continuous}
Woods, D.~C. and Lewis, S.~M. (2011), \enquote{Continuous optimal designs for
  generalized linear models under model uncertainty,} \textit{Journal of
  Statistical Theory and Practice}, 5, 137--145.

\bibitem[{Woods et~al.(2006)Woods, Lewis, Eccleston, and
  Russell}]{woods2006designs}
Woods, D.~C., Lewis, S.~M., Eccleston, J.~A., and Russell, K. (2006),
  \enquote{Designs for generalized linear models with several variables and
  model uncertainty,} \textit{Technometrics}, 48, 284--292.

\bibitem[{Woods et~al.(2017)Woods, Overstall, Adamou, and
  Waite}]{woods2017bayesian}
Woods, D.~C., Overstall, A.~M., Adamou, M., and Waite, T.~W. (2017),
  \enquote{Bayesian design of experiments for generalized linear models and
  dimensional analysis with industrial and scientific application,}
  \textit{Quality Engineering}, 29, 91--103.

\bibitem[{Wu and Stufken(2014)}]{14wu}
Wu, H.-P. and Stufken, J. (2014), \enquote{{Locally $\phi_p$-optimal designs
  for generalized linear models with a single-variable quadratic polynomial
  predictor},} \textit{Biometrika}, 101, 365--375.

\bibitem[{Yang et~al.(2013)Yang, Biedermann, and Tang}]{13yang}
Yang, M., Biedermann, S., and Tang, E. (2013), \enquote{On Optimal Designs for
  Nonlinear Models: A general and Efficient Algorithm,} \textit{Journal of the
  American Statistical Association}, 108, 1411--1420.

\bibitem[{Yang and Stufken(2009)}]{09yang}
Yang, M. and Stufken, J. (2009), \enquote{Support points of locally optimal
  designs for nonlinear models with two parameters,} \textit{The Annals of
  Statistics}, 37, 518--541.

\bibitem[{Yang and Stufken(2012)}]{12yang}
--- (2012), \enquote{Identifying Locally Optimal Designs for Nonlinear Models:
  A Simple Extension with Profound Consequences,} \textit{The Annals of
  Statistics}, 40, 1665--1681.

\bibitem[{Yang et~al.(2011)Yang, Zhang, and Huang}]{yang2011optimal}
Yang, M., Zhang, B., and Huang, S. (2011), \enquote{Optimal designs for
  generalized linear models with multiple design variables,} \textit{Statistica
  Sinica}, 1415--1430.

\bibitem[{Yu(2010)}]{10yu}
Yu, Y. (2010), \enquote{Monotonic Convergence of a General Algorithm for
  Computing Optimal Designs,} \textit{The Annals of Statistics}, 38,
  1593--1606.

\end{thebibliography}

\vskip .65cm
\noindent
Department of Mathematical Sciences, DePaul University, Chicago, IL
\vskip 2pt
\noindent
E-mail: yli139@depaul.edu
\vskip 2pt

\noindent
Department of Applied Mathematics, Illinois Institute of Technology, Chicago, IL
\vskip 2pt
\noindent
E-mail: lkang2@iit.edu

\noindent
Department of Statistics, Virginia Tech, Blacksburg, VA
\vskip 2pt
\noindent
E-mail: xdeng@vt.edu

\newpage
\setcounter{page}{1}
\section{Appendix}
\setcounter{lem}{0}
\begin{lem}\label{lem:SeConvex}
$\se(\cdot,\modelspace)$ defined in \eqref{eqn:criterion} is a convex function of design $\xi$.
\end{lem}
\begin{proof}
Since $\Phi_p(\cdot)$ is convex with respect to $\xi$ \citep{12yang} and $\exp(\cdot)$ is a convex and strictly increasing function, the composite function $\xi\mapsto\exp\left(\frac{\Phi_p^j(\xi)}{\Phi_p^{\opt_j}}\right)$ is a convex function of $\xi$. As a result,  $\se(\cdot,\modelspace)$ is a convex function of $\xi$.
\end{proof}

\begin{lem}\label{lem:dirder}
The directional derivative of $\se(\xi,\modelspace)$ in the direction from $\xi$ to $\xi'$ is
\begin{align*}
\text{when }p=0, & \quad \phi(\xi',\xi)=\sum\limits_{j=1}^{m}\tphi^j_0(\xi)\left[q-\tr\left(\mF_j(\xi)^{-1}\mG_j(\xi,\xi')\right)\right],\\
\text{when }p>0, & \quad \phi(\xi',\xi)=\sum\limits_{j=1}^{m}\tphi^j_p(\xi)\left[\Phi_p^j(\xi)-q^{-1/p}\left(\tr\left(\mF_j(\xi)\right)^p\right)^{1/p-1} \tr\left(\left(\mF_j(\xi)\right)^{p-1}\mG_j(\xi,\xi')\right)\right],
\end{align*}
where
\begin{align*}
\tphi^j_p(\xi) &= \left[\Phi_p^{\opt_j}\right]^{-1}\exp\left(\frac{\Phi_p^j(\xi)}{\Phi_p^{\opt_j}}\right),\quad
\mB_j = \left.\frac{\partial \vf(\vbeta)}{\partial \vbeta^{\top}}\right|_{\vbeta = \vbeta_j},\\
\mF_j(\xi) &= \mB_j\mI_j(\xi)^{-1}\mB_j^{\top}, \quad \mG_j(\xi,\xi') = \mB_j\mI_j(\xi)^{-1}\mI_j(\xi')\mI_j(\xi)^{-1}\mB_j^{\top}.
\end{align*}
\end{lem}
\begin{proof}
Given $\tilde{\xi} = (1-\alpha)\xi+\alpha\xi'$, we have $\mI_j(\tilde{\xi}) = (1-\alpha)\mI_j(\xi)+\alpha\mI_j(\xi')$.
For any invertible matrix $\mS$, whose elements are functions of $\alpha$, the derivative of $\mS^{-1}$ is $\frac{\partial \mS^{-1}}{\partial \alpha} = -\mS^{-1}\frac{\partial \mS}{\partial \alpha}\mS^{-1}$.
So, the derivative of $\mI_j(\tilde{\xi})^{-1}$ with respect to $\alpha$ can be expressed as,
\begin{equation}\label{eqn:derder}
\frac{\partial \left[\mI_j(\tilde{\xi})^{-1}\right]}{\partial\alpha} = -\mI_j(\tilde{\xi})^{-1}[\mI_j(\xi')-\mI_j(\xi)]\mI_j(\tilde{\xi})^{-1}.
\end{equation}
Thus, for $p=0$,
\begin{align}\nonumber
\frac{\partial \Phi_0^j(\tilde{\xi})}{\partial\alpha} &=\frac{\partial \log\left|\mF_j(\tilde{\xi})\right|}{\partial\alpha}
=  \tr\left[\mF_j(\tilde{\xi})^{-1}\mB_j\frac{\partial \left[\mI_j(\tilde{\xi})^{-1}\right]}{\partial\alpha}\mB_j^{\top} \right]\\
&= -\tr\left[\mF_j(\tilde{\xi})^{-1}\mB_j\mI_j(\tilde{\xi})^{-1}[\mI_j(\xi')-\mI_j(\xi)]\mI_j(\tilde{\xi})^{-1}\mB_j^{\top} \right];\label{eqn:phider0}
\end{align}
for $p>0$,
\begin{align}\nonumber
\frac{\partial \Phi_p^j(\tilde{\xi})}{\partial\alpha}&= \frac{\partial \left[\left(q^{-1}\tr\left(\mF_j(\tilde{\xi})\right)^{p}\right)^{1/p}\right]}{\partial\alpha}\\
& = -q^{-1/p}\left(\tr\left(\mF_j(\tilde{\xi})\right)^p\right)^{1/p-1}\tr\left[\left(\mF_j(\tilde{\xi})\right)^{p-1}\mB_j\mI_j(\tilde{\xi})^{-1}[\mI_j(\xi')-\mI_j(\xi)]\mI_j(\tilde{\xi})^{-1}\mB_j^{\top}\right].\label{eqn:phider}
\end{align}
Based on \eqref{eqn:phider0} and \eqref{eqn:phider}, the directional derivative of $\se(\xi,\modelspace)$ is
\begin{align*}
\phi(\xi',\xi) &= \left.\frac{\partial\left[ \sum\limits_{j=1}^m \exp\left(\frac{\Phi_p^j(\tilde{\xi})}{\Phi_p^{\opt_j}}\right)\right]}{\partial\alpha}\right|_{\alpha=0}
= \sum\limits_{j=1}^{m}\left.\tphi^j_p(\xi)\frac{\partial \Phi_p^j(\tilde{\xi})}{\partial\alpha}\right|_{\alpha=0}\\
&=\left\{\begin{array}{ll}
\sum\limits_{j=1}^{m}\tphi^j_0(\xi)\left[q-\tr\left(\mF_j(\xi)^{-1}\mG_j(\xi,\xi') \right)\right],\,\,\,& p = 0;\\
\sum\limits_{j=1}^{m}\tphi^j_p(\xi)\left[\Phi_p^j(\xi)-q^{-1/p}\left(\tr\left(\mF_j(\xi)\right)^p\right)^{1/p-1} \tr\left(\left(\mF_j(\xi)\right)^{p-1}\mG_j(\xi,\xi')\right)\right], & p>0. \end{array}\right.
\end{align*}
\end{proof}

\begin{lem}\label{lem:dirder2}
The directional derivative of $\se(\xi,\modelspace)$ in the direction of a single point $\vx$ is given as,
$$\phi(\vx,\xi) = \left\{\begin{array}{ll}
\sum\limits_{j=1}^{m}\tphi^j_p(\xi)\left[q-w_j(\vx)\vg_j^{\top}(\vx)\mM_j(\xi)\vg_j(\vx)\right],\,\,\,& p = 0;\\
\sum\limits_{j=1}^{m}\tphi^j_p(\xi)\left[\Phi_p^j(\xi)-q^{-1/p}w_j(\vx)\left(\tr\left(\mF_j(\xi)\right)^p\right)^{1/p-1}\vg_j^{\top}(\vx)\mM_j(\xi)\vg_j(\vx)\right],\,\,\,& p > 0,\end{array}\right.$$
where $\mM_j(\xi) = \mI_j(\xi)^{-1}\mB_j^{\top}\mF_j(\xi)^{p-1}\mB_j\mI_j(\xi)^{-1}$.
\end{lem}

Particularly, the directional derivatives of D-, A- and the prediction-oriented criterion EI-optimality defined in \citep{li2018efficient} are:
$$\phi(\vx,\xi) = \left\{\begin{array}{ll}
\sum\limits_{j=1}^{m}\tphi^j_p(\xi)\left[l-w_j(\vx)\vg_j^{\top}(\vx)\mI^{-1}_j(\xi)\vg_j(\vx)\right],\,\,\,&\text{D-optimality};\\
\sum\limits_{j=1}^{m}\tphi^j_p(\xi)\left[\tr(\mI_j^{-1}(\xi))-w_j(\vx)\vg_j^{\top}(\vx)\mI^{-2}_j(\xi)\vg_j(\vx)\right],\,\,\,&\text{A-optimality};\\
\sum\limits_{j=1}^{m}\tphi^j_p(\xi)\left[\tr(\mI_j^{-1}(\xi)\mA_j)-w_j(\vx)\vg_j^{\top}(\vx)\mI^{-1}_j(\xi)\mA_j\mI^{-1}_j(\xi)\vg_j(\vx)\right],\,\,\,&\text{EI-optimality},
\end{array}\right.$$
where $\mA_j =\int_{\Omega}\vg_j(\vx)\vg_j^\top(\vx)\left[\frac{\dif h_j^{-1}}{\dif \eta_j}\right]^2\dif F_{\IMSE}(\vx)$ for EI-optimality is pre-determined and it does not depend on the design $\xi$.
The cdf $F_{\IMSE}$ is the user-specified distribution for the EI-optimality.

\begin{lem}\label{lem:ConvWeights}
For a design $\xi^{\vlambda}$ with fixed design points $\vx_1,\cdots,\vx_n$, $\se(\cdot,\modelspace)$ in \eqref{eqn:weightlb} is a convex function with respect to the weight vector $\vlambda$.
\end{lem}
\begin{proof}
The proof is similar to that of Lemma \ref{lem:SeConvex}.
\end{proof}

\noindent\textbf{Proof of Theorem \ref{thm:equi_thm}}
\begin{proof}
\begin{enumerate}[i.~]
\item
$(1)\rightarrow (2)$: As $\se(\cdot,\modelspace)$ is a convex function in $\xi$ proved in Lemma \ref{lem:SeConvex}, the directional derivative $\phi(\vx,\robdes)\geq 0$ holds for any $\vx\in\Omega$, and the inequality becomes equality if $\vx$ is a support point of the design $\robdes$.
\item
$(2) \rightarrow (1)$: If $\phi(\vx,\robdes)\geq 0$ holds for any $\vx\in\Omega$, then $\robdes$ minimizes $\se(\xi,\modelspace)$ as  $\se(\cdot,\modelspace)$ is a convex function in $\xi$.
\end{enumerate}
\end{proof}

\noindent\textbf{Proof of Lemma \ref{lem:ineqofdirder}}
\begin{proof}
The lemma is proved for the $p>0$ case, and the case $p=0$ could be proved similarly.
The directional derivative of $\se(\xi,\modelspace)$ in the direction of the point $\vx^{*} = \argmin\limits_{\vx\in\Omega}\phi(\vx,\xi)$ is:
\begin{eqnarray}\label{inequ:phixxi}
& &\min\limits_{\vx\in\Omega} \phi(\vx,\xi)=\phi(\vx^{*},\xi)\nonumber\\
 &=& \sum\limits_{j=1}^{m}\tphi_j(\xi)\left[\Phi_p^j(\xi)-q^{-1/p}\left(\tr\left(\mF_j(\xi)\right)^p\right)^{1/p-1} \tr\left(\left(\mF_j(\xi)\right)^{p-1}\mB_j\mI_j(\xi^{(r)})^{-1}\mI_j(\vx_r^*)\mI_j(\xi)^{-1}\mB_j^{\top}\right)\right]\nonumber\\
 &\leq & \sum\limits_{j=1}^{m}\tphi_j(\xi)\left[\Phi_p^j(\xi)-q^{-1/p}\left(\tr\left(\mF_j(\xi)\right)^p\right)^{1/p-1} \tr\left(\left(\mF_j(\xi)\right)^{p-1}\mB_j\mI_j(\xi)^{-1}\mI_j(\vz)\mI_j(\xi)^{-1}\mB_j^{\top}\right)\right]\nonumber\\
\end{eqnarray}
for any $\vx\in\Omega$, where $\mI_j(\vx)$ denotes the information matrix of the design with a unit mass on single point $\vx$.

Denote the Mm-$\Phi_p$ design $\robdes = \left\{\begin{array}{ccc}
\vx_1,&...,&\vx_n\\
\lambda^*_1,&...,&\lambda^*_n
\end{array}\right\}$. With \eqref{inequ:phixxi}, we have
\begin{eqnarray}\label{inequ:phixixi}
&&\phi(\vx^*,\xi)\nonumber\\
 &\leq & \sum_{i=1}^n\lambda^*_i\sum\limits_{j=1}^{m}\tphi_j(\xi)\left[\Phi_p^j(\xi)-q^{-1/p}\left(\tr\left(\mF_j(\xi)\right)^p\right)^{1/p-1} \tr\left(\left(\mF_j(\xi)\right)^{p-1}\mB_j\mI_j(\xi)^{-1}\mI_j(\vx_i)\mI_j(\xi)^{-1}\mB_j^{\top}\right)\right]\nonumber\\
 &=&\sum\limits_{j=1}^{m}\tphi_j(\xi)\left[\Phi_p^j(\xi)-q^{-1/p}\left(\tr\left(\mF_j(\xi)\right)^p\right)^{1/p-1} \tr\left(\left(\mF_j(\xi)\right)^{p-1}\mB_j\mI_j(\xi)^{-1}\mI_j(\robdes)\mI_j(\xi)^{-1}\mB_j^{\top}\right)\right]\nonumber\\
 &=& \phi(\robdes,\xi)
\end{eqnarray}

Furthermore, with the definition of directional derivative in the direction of the Mm-$\Phi_p$ design $\robdes$ and convexity of $\se(\cdot,\modelspace)$, we have
\begin{eqnarray}\label{eqn:dirinequ}
\phi(\robdes,\xi) &=& \lim_{\alpha\rightarrow 0}\frac{\se((1-\alpha)\xi+\alpha\robdes,\modelspace)-\se(\xi,\modelspace)}{\alpha}\nonumber\\
&\leq &\lim_{\alpha\rightarrow 0}\frac{(1-\alpha)\se(\xi,\modelspace)+\alpha\se(\robdes,\modelspace)-\se(\xi,\modelspace)}{\alpha}\nonumber\\
&=& \se(\robdes,\modelspace)-\se(\xi,\modelspace)
\end{eqnarray}

Combining \eqref{inequ:phixixi} and \eqref{eqn:dirinequ}, we complete the proof that
$$\min\limits_{\vx\in\Omega} \phi(\vx,\xi)\leq \phi(\robdes,\xi)\leq \se(\robdes,\modelspace)-\se(\xi,\modelspace)\leq 0.$$
\end{proof}

\noindent\textbf{Proof of Theorem \ref{thm:lowerboundeff}}
\begin{proof}
When $1+2\frac{\min\limits_{\vx\in\Omega} \phi(\vx,\xi)}{\se(\xi,\modelspace)}<0$, $\Eff_{\lse}(\xi,\robdes;\modelspace) \geq  1+2\frac{\min\limits_{\vx\in\Omega} \phi(\vx,\xi)}{\se(\xi,\modelspace)}$ holds automatically.

When $1+2\frac{\min\limits_{\vx\in\Omega} \phi(\vx,\xi)}{\se(\xi,\modelspace)}\geq 0$, that is, $\frac{\min\limits_{\vx\in\Omega} \phi(\vx,\xi)}{\se(\xi,\modelspace)}\geq -0.5$, define $\frac{\se(\robdes,\modelspace)}{\se(\xi,\modelspace)} = a >0$, then it follows immediately from Lemma \ref{lem:ineqofdirder} that
\begin{equation}\label{inequ:eff}
1\geq a = \frac{\se(\robdes,\modelspace)}{\se(\xi,\modelspace)}\geq 1+\frac{\min\limits_{\vx\in\Omega} \phi(\vx,\xi)}{\se(\xi,\modelspace)}\geq 0.5.
\end{equation}

Since the function $\frac{\ln(a)}{\ln(\se(\xi,\modelspace))}+1-a$ is an increasing function of $\se(\xi,\modelspace)$, $\se(\xi,\modelspace)\geq e$ because of its definition and $a\geq 0.5$, we have
\begin{eqnarray*}
&&\left|\Eff_{\lse}(\xi,\robdes;\modelspace)-\frac{\se(\robdes,\modelspace)}{\se(\xi,\modelspace)}\right|= \left|\frac{\ln(a\se(\xi,\modelspace))}{\ln(\se(\xi,\modelspace))}-\frac{a\se(\xi,\modelspace)}{\se(\xi,\modelspace)}\right|\\
&=& \left|\frac{\ln(a)}{\ln(\se(\xi,\modelspace))}+1-a\right|\leq \max(\left|\ln(a)+1-a\right|,\left|1-a\right|)\\
&=& \max(-\ln(a)-1+a,1-a)=1-a.
\end{eqnarray*}
Thus, together with \eqref{inequ:eff}, $\Eff_{\lse}(\xi,\robdes;\modelspace)\geq 2 a-1\geq 1+2\frac{\min\limits_{\vx\in\Omega} \phi(\vx,\xi)}{\se(\xi,\modelspace)}$.
\end{proof}

\noindent\textbf{Proof of Theorem \ref{thm:cong-algo2}}
\begin{proof}
We show the proof for the scenario $p>0$ in the $\Phi_p$-criterion. The proof for $p=0$ could be done similarly.
The proof is established by proof of contradiction.
Suppose that the Algorithm 1 does not converge to the Mm-$\Phi_p$ design $\xi^*$, then we have
$$\lim_{r\rightarrow\infty} \se(\xi^{(r)},\modelspace) > \se(\robdes,\modelspace).$$

For any iteration $r+1\geq 1$, since $\mathcal{X}^{(r)}\subset \mathcal{X}^{(r+1)}$ and the Optimal-Weight Procedure returns optimal weight vector that minimizes $\se$ criterion,
the design $\xi^{(r+1)}$ cannot be worse than the design in the previous iteration $\xi^{(r)}$, i.e.,
$$\se(\xi^{(r+1)},\modelspace)\leq \se(\xi^{(r)},\modelspace).$$
Thus, for all $r\geq 0$, there exists $a>0$, such that,
$$\se(\xi^{(r)},\modelspace)>\se(\robdes,\modelspace)+a.$$
According to Lemma \ref{lem:ineqofdirder},
$$\phi(\vx_r^*,\xi^{(r)})\leq \phi(\robdes,\xi^{(r)})\leq \se(\robdes,\modelspace)-\se(\xi^{(r)},\modelspace)<-a,$$
for any $r\geq 0$.
Then, the Taylor expansion of $\se((1-\alpha)\xi^{(r)}+\alpha\vx_r^*,\modelspace)$ is upper bounded by
\begin{eqnarray}\label{eqn:taylorlb}
\se((1-\alpha)\xi^{(r)}+\alpha\vx_r^*,\modelspace) &=& \se(\xi^{(r)},\modelspace)+\phi(\vx_r^*,\xi^{(r)})\alpha+\frac{u}{2}\alpha^2\nonumber\\
&<& \se(\xi^{(r)},\modelspace)-a\alpha+\frac{u}{2}\alpha^2,
\end{eqnarray}
where $u\geq 0$ is the second-order directional derivative of $\se$ evaluated at a value between 0 and $\alpha$.

For Algorithm 1, the criterion $\se$ is minimized, for all $0\leq\alpha\leq 1$ we have
\begin{eqnarray*}
\se(\xi^{(r+1)},\modelspace) &\leq& \se((1-\alpha)\xi^{(r)}+\alpha\vx_r^*,\modelspace)\\
&<& \se(\xi^{(r)},\modelspace)-a\alpha+\frac{u}{2}\alpha^2,
\end{eqnarray*}
or equivalently,
\begin{eqnarray*}
&&\se(\xi^{(r+1)},\modelspace) - \se(\xi^{(r)},\modelspace) < -a\alpha+\frac{u}{2}\alpha^2 = \frac{u}{2}\left(\alpha-\frac{a}{u}\right)^2-\frac{a^2}{2u}\\
&<&\left\{\begin{array}{ll}-\frac{a^2}{2u}<0,\,\,\,\,&\text{choosing}\,\,\alpha = \frac{a}{u}\,\,\text{if}\,\,\,\,\,a\leq u\\
\frac{u-4a}{8}<0,\,\,\,\,&\text{choosing}\,\,\alpha = 0.5\,\,\text{if}\,\,\,\,\,a> u
\end{array}\right..
\end{eqnarray*}
As a result, $\lim\limits_{r\rightarrow \infty}\se(\xi^{(r)},\modelspace)  = -\infty$, which contradicts with the fact that $\se(\xi^{(r)},\modelspace)\geq 0$ for any design $\xi^{(r)}$. Thus,
$$\lim\limits_{r\rightarrow\infty} \se(\xi^{(r)},\modelspace) = \se(\robdes,\modelspace).$$
Since $\ln(\cdot)$ on $[1,\infty)$ is a continuous function,
$$\lim\limits_{r\rightarrow\infty} \lse(\xi^{(r)},\modelspace) = \lse(\robdes,\modelspace).$$
\end{proof}

\noindent\textbf{Description of Algorithm \ref{alg:weight}}. 
\begin{algorithm}[ht]
  \caption{(\textbf{Optimal-Weight Procedure}) A Modified Multiplicative Approach. \label{alg:weight}}
  \begin{algorithmic}[1]
  \State Assign a uniform initial weight vector $\vlambda^{(0)} = [\lambda_1^{(0)},...,\lambda_n^{(0)}]^{\top}$, and $k=0$.
  \While {$change>Tol$ and $k<MaxIter_2$}
  \For {$i = 1, \ldots, n$}
 \State Update the weight of design point $\vx_i$:
\begin{eqnarray}\label{for:multialg}
&& \lambda_i^{(k+1)} = \lambda_i^{(k)} \frac{\left(d_p(\vx_i,\xi^{\vlambda^{(k)}})\right)^{\delta}}{\sum\limits_{s=1}^n\lambda_s^{(k)}\left(d_p(\vx_s,\xi^{\vlambda^{(k)}})\right)^{\delta}},\nonumber\\
&=&  \left\{\begin{array}{ll}
 \lambda_i^{(k)} \frac{\left(\sum\limits_{j=1}^m \tphi_0^j(\xi^{\vlambda^{(k)}}) w_j(\vx_i)\vg_j^{\top}(\vx_i)\mM_j(\xi^{\vlambda^{(k)}})\vg_j(\vx_i)\right)^\delta}{\sum\limits_{s=1}^n \lambda_s^{(k)}\left(\sum\limits_{j=1}^m \tphi_0^j(\xi^{\vlambda^{(k)}}) w_j(\vx_i)\vg_j^{\top}(\vx_i)\mM_j(\xi^{\vlambda^{(k)}})\vg_j(\vx_i)\right)^\delta}, & p=0;\\
   \lambda_i^{(k)} \frac{\left(\sum\limits_{j=1}^m\tphi_p^j(\xi^{\vlambda^{(k)}})w_j(\vx_i)\left(\tr\left(\mF_j(\xi^{\vlambda^{(k)}})\right)^p\right)^{1/p-1}\vg_j^{\top}(\vx_i)\mM_j(\xi^{\vlambda^{(k)}})\vg_j(\vx_i)\right)^\delta}{\sum\limits_{s=1}^n \lambda_s^{(k)}\left(\sum\limits_{j=1}^m\tphi_p^j(\xi^{\vlambda^{(k)}})w_j(\vx_i)\left(\tr\left(\mF_j(\xi^{\vlambda^{(k)}})\right)^p\right)^{1/p-1}\vg_j^{\top}(\vx_i)\mM_j(\xi^{\vlambda^{(k)}})\vg_j(\vx_i)\right)^\delta }, & p>0.
 \end{array}\right.
\end{eqnarray}
\State $change=\max\limits_{i=1,\cdots,n}(|\lambda_i^{(k+1)}-\lambda_i^{(k)}|)$.
\State $k=k+1$.
\EndFor
 \EndWhile
\end{algorithmic}
\end{algorithm}

There are three user-specified parameters $\delta$, $Tol$, and $MaxIter_2$ in Algorithm \ref{alg:weight}.
$Tol$ is the tolerance of convergence, and we usually set it to be $Tol=1e-15$.
$MaxIter_2$ is the maximum number of iterations and we set it to be $MaxIter_2=200$ in all numerical examples.
The parameter $\delta\in (0,1]$ plays the same role as in the classical multiplicative algorithm \citep{78silvey}, which is to control the speed of the convergence.
According to the numerical study by \cite{74fellman} and \cite{83tor}, $\delta$ is often chosen as 1 for D-optimality, and 0.5 for A- or EI-optimality.

Derivation of \eqref{for:multialg} is stated as follows.
Update the weight of $\vx_i$ in iteration $k$ with
\begin{equation}\label{for:multialg1}
\tilde{\lambda}_i^{(k+1)} = \lambda_i^{(k)} \left(\frac{d_p(\vx_i,\xi^{\vlambda^{(k)}})}{\sum\limits_{t=1}^n\lambda_t^{(k)}d_p(\vx_t,\xi^{\vlambda^{(k)}})}\right)^{\delta},
\end{equation}
then normalize the weights to ensure the sum 1 condition as
\begin{equation}\label{for:multialg2}
\lambda_i^{(k+1)} = \frac{\tilde{\lambda}_i^{(k+1)}}{\sum\limits_{s=1}^n \tilde{\lambda}_s^{(k+1)}}.
\end{equation}
Plugging \eqref{for:multialg1} into \eqref{for:multialg2}, we have $\lambda_i^{(k+1)} =  \lambda_i^{(k)} \frac{\left(d_p(\vx_i,\xi^{\vlambda^{(k)}})\right)^{\delta}}{\sum\limits_{s=1}^n\lambda_s^{(k)}\left(d_p(\vx_s,\xi^{\vlambda^{(k)}})\right)^{\delta}}.$

\noindent\textbf{Additional Tables and Figures for Section \ref{sec: potato}}.
\begin{table}[htbp]
  \centering
  \caption{Model Space $\mathcal{M}$ of Potato Packing Example}
    \begin{tabular}{|lrrr|}
    \hline
    Term  & First-Order $M_1$ & With interaction $M_2$ & Second-order $M_3$  \\
    \hline
    Intercept & -0.28 & -1.44 & -2.93 \\
     $x_1$     & 0     & 0     & 0 \\
       $x_2$   & -0.76 & -1.95 & -0.52 \\
     $x_3$     & -1.15 & -2.36 & -0.79 \\
     $x_1x_2$     &       & 0     & 0 \\
      $x_1x_3$     &       & 0     & 0 \\
      $x_2x_3$     &       & -2.34 & -0.66 \\
     $x_1^2$     &       &       & 0.94 \\
     $x_2^2$     &       &       & 0.79 \\
     $x_3^2$     &       &       & 1.82 \\
    \hline
    \end{tabular}%
  \label{tab:PotatoPackModel}%
\end{table}%

\begin{figure}[ht]
\centering
\includegraphics[width=12cm]{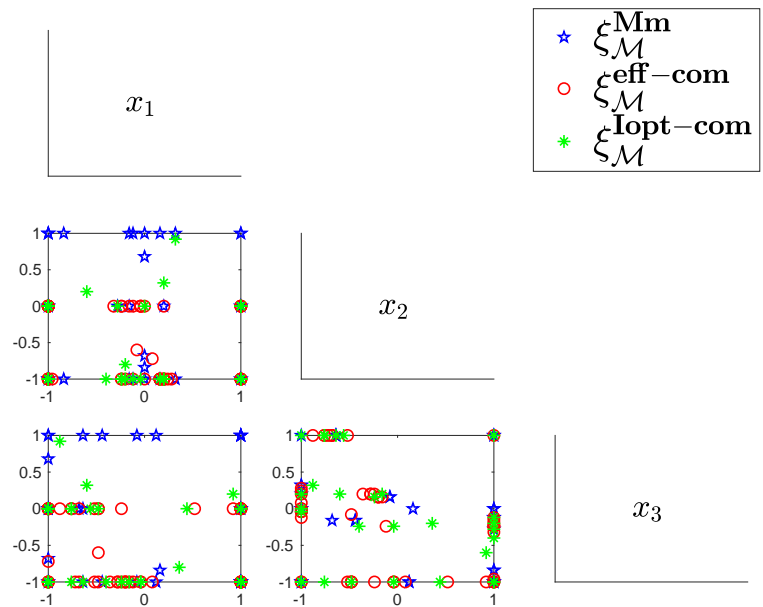}
\caption{Design Points of Mm-$\Phi_p$ Design and Compromise Designs}
\label{fig:PotatoPackPoints}
\end{figure}



\end{document}